\documentclass[11pt]{article}
\usepackage[numbers]{natbib}
\usepackage{xspace}
\usepackage[usenames,dvipsnames]{xcolor}
\usepackage{graphicx}
\usepackage{amsmath,amssymb,amsthm}
\usepackage[margin=1in]{geometry}
\usepackage[algoruled,noend]{algorithm2e}
\usepackage{pstricks}
\usepackage{comment}
\usepackage{dsfont}
\usepackage{tikz}
\usepackage{tikz-dependency}

\usepackage{thmtools} 
\usepackage{thm-restate}

\usepackage[colorlinks=true,urlcolor=blue,linkcolor=RoyalBlue,citecolor=OliveGreen]{hyperref}

\newtheorem{theorem}{Theorem}
\newtheorem{inftheorem}{Informal Theorem}

\newtheorem{claim}{Claim}
\newtheorem{lemma}{Lemma}
\newtheorem{corollary}{Corollary}
\newtheorem{obs}{Observation}
\newtheorem{definition}{Definition}
\newtheorem{example}{Example}

\begin{document}
\title{On the Competition Complexity of Dynamic Mechanism Design}
\date{}

\author{
Siqi Liu
\\ UC Berkeley
\\ \href{mailto:sliu18@berkeley.edu}{sliu18@berkeley.edu}
\and
Christos-Alexandros Psomas 
\\ Carnegie Mellon University
\\ \href{mailto:cpsomas@cs.cmu.edu}{cpsomas@cs.cmu.edu}
}

\newcommand{\days}{m}
\newcommand{\myerson}[1]{\textsc{Mye} \left[ #1 \right]}
\newcommand{\opt}[1]{\textsc{OPT} \left[ #1 \right]}
\newcommand{\expectation}[1]{\mathbb{E} \left[ #1 \right]}
\newcommand{\rev}[1]{\textsc{Rev} \left[ #1 \right]}
\newcommand{\downnode}[2]{#1_{#2 , \leftarrow}}
\newcommand{\upnode}[2]{#1_{#2 , \rightarrow} }

\newenvironment{talign}
 {\let\displaystyle\textstyle\align}
 {\endalign}
\newenvironment{talign*}
 {\let\displaystyle\textstyle\csname align*\endcsname}
 {\endalign}

\maketitle

\begin{abstract}

The \textit{Competition Complexity} of an auction measures how much competition is needed for the revenue of a simple auction to surpass the optimal revenue. A classic result from auction theory by Bulow and Klemperer~\cite{bulow1996auctions}, states that the Competition Complexity of VCG, in the case of $n$ i.i.d. buyers and a single item, is $1$. In other words, it is better to invest in recruiting one extra buyer and run a second price auction than to invest in learning \textit{exactly} the buyers' underlying distribution and run the revenue-maximizing auction \textit{tailored} to this distribution.

In this paper we study the Competition Complexity of \textit{dynamic auctions}. Consider the following problem: a monopolist is auctioning off $\days$ items in $\days$ consecutive stages to $n$ interested buyers. A buyer realizes her value for item $k$ in the beginning of stage $k$. \emph{How many additional buyers are necessary and sufficient for a second price auction at each stage to extract revenue at least that of the optimal dynamic auction? }
We prove that the Competition Complexity of dynamic auctions is at most $3n$ - and at least linear in $n$ - even when the buyers' values are correlated across stages, under a monotone hazard rate assumption on the stage (marginal) distributions. This assumption can be relaxed if one settles for independent stages.
We also prove results on the number of additional buyers necessary for VCG at every stage to be an $\alpha$-approximation of the optimal revenue; we term this number the $\alpha$-\textit{approximate Competition Complexity}. For example, under the same mild assumptions on the stage distributions we prove that one extra buyer suffices for a $\frac{1}{e}$-approximation. 
As a corollary we provide the first results on \textit{prior-independent} dynamic auctions.
This is, to the best of our knowledge, the first non-trivial positive guarantees for simple ex-post IR dynamic auctions for \textit{correlated} stages.

A key step towards proving bounds on the Competition Complexity is getting a good benchmark/upper bound to the optimal revenue.
To this end, we extend the recent duality framework of~\citet{cai2016} to dynamic settings.
As an aside to our approach we obtain a revenue non-monotonicity lemma for dynamic auctions, which may be of independent interest. 

\end{abstract}

\thispagestyle{empty}

\clearpage
\setcounter{page}{1}

\section{Introduction}


A monopolist is auctioning off $\days$ items in $\days$ consecutive stages to $n$ interested buyers. A buyer learns her value for the $k$-th item at the beginning of stage $k$. For the future items only a prior distribution is known. The prior of buyer $i$ depends on the values of that buyer for the items so far. How should this monopolist behave in order to maximize her expected revenue?  
The optimal dynamic mechanism can be extremely complex, both computationally and in terms of its description. Even worse, the optimal mechanism depends on detailed knowledge of the buyers' distributions (across time) in intricate ways. This paper aims to answer the following question: {\em Can we design simple dynamic mechanisms that do not depend on details of the underlying distributions?} 

We are not the first to face such problems. Optimal mechanisms can be extremely complex even for a static auction with a single additive buyer (e.g. $m$ uniform i.i.d. items on $[c,c+1]$ as observed by~\citet{daskalakis2015strong}). But even for the case of $n$ i.i.d. buyers and a single item for sale, where Myerson's theory~\cite{myerson1981optimal} readily applies, the necessity of good prior information makes the solution somewhat less appealing in practice. Can this be avoided?
An elegant result from auction theory, by Bulow and Klemperer~\cite{bulow1996auctions}, states that the revenue of a second price auction with $n + 1$ buyers with valuations drawn i.i.d. from a regular distribution\footnote{A distribution $D$ with density $f$ and cumulative density $F$ is regular if Myerson's virtual function $\phi(v) = v - \frac{1-F(v)}{f(v)}$ is monotone non-decreasing.} is at least that of the optimal auction, \textit{tailored} for the exact distributions, for $n$ buyers. In other words, it is better to invest in recruiting one more agent and run VCG, than to invest in \textit{exactly} learning the underlying value distribution and then run the revenue-maximizing auction tailored to this distribution. One of the reasons this theorem is so appealing is that VCG is prior-independent, meaning its description is independent of the underlying distribution. The Bulow-Klemperer theorem can also be seen as a ``resource augmentation'' argument. The optimal auction extracts more revenue than a second price auction by definition; their theorem gives an intuitive characterization of how much.


In this paper we prove the first Bulow-Klemperer type results for dynamic auctions. We are interested in the number of additional buyers necessary (these extra buyers will be present in all stages) for a second price auction \textit{at every stage} to have expected revenue at least that of the optimal dynamic auction.
We adopt the terminology of~\citet{eden2016competition} and call this number the \textit{Competition Complexity}. We also define the $\alpha$-\textit{approximate Competition Complexity} to be the extra number of buyers necessary for a second price auction at each stage to be an $\alpha$ approximation (in terms of revenue) of the optimal dynamic auction\footnote{We use the term $\alpha$-Competition Complexity to have a single language when stating our results. As we see later though, strictly outperforming $OPT$ is a much more challenging task than getting an $\alpha$ approximation.}.
Thus, in this terminology, Bulow and Klemperer's theorem says that the Competition Complexity of a single item static auction with $n$ buyers is $1$.

The original result of Bulow-Klemperer stems from our deep understanding of revenue optimal auctions in the single item case, thanks to Myerson's work. In contrast, maximizing revenue in dynamic or other multi-dimensional environments is poorly understood. Moreover, known gaps (\cite{papadimitriou2016complexity}) between the revenue of adaptive and non-adaptive dynamic auctions imply that assumptions stronger than regularity of the stage distributions  will be necessary.
Despite these obstacles, for {\em correlated stages} we can show the following bounds:

\begin{inftheorem}\label{thm: approx competition complexity}
In the case of $n$ buyers and $\days$ correlated stages (the value of buyer $i$ on stage $k$ can be correlated with past and future values, but not with other buyers' values), if all the stage (marginal) distributions have monotone hazard rate (MHR)\footnote{A distribution has \textit{monotone hazard rate} (MHR) if its hazard rate $h(v) = \frac{f(v)}{1-F(v)}$ is monotone non-decreasing.}: (1) The Competition Complexity is at most $3n$ and at least $(e-1)n$; (2) The $\frac{1}{e}$-approximate Competition Complexity is $1$; (3) The $\frac{1}{3}$-approximate Competition Complexity is $0$.
\end{inftheorem}

In other words, if the stage distributions have monotone hazard rate, recruiting $3n$ additional buyers is {\em strictly} better than learning all $\days$ stage distributions exactly, including the possible correlation between stages. If an approximation of $\frac{1}{e}$ suffices, only one additional buyer is necessary. Even more surprisingly, simply running a second price auction at each stage extracts a $\frac{1}{3}$-approximation of the optimal revenue. These bounds hold even against ex-ante IR dynamic auctions. In other words, running a second price auction at each stage, which requires no distributional knowledge whatsoever, is a $\frac{1}{3}$ approximation to the impossible benchmark of an optimal, ex-ante IR auction that uses full knowledge of the buyers' arbitrarily correlated distributions!

Many common families of distributions such as the Uniform, Exponential and Normal have MHR. 
To get the theorem we prove new bounds on the expected order statistics of MHR distributions and combine them with an upper bound (for the optimal dynamic revenue) equal to the social welfare, i.e. the sum (over stages) of the expected first order statistic $\expectation{(X_k)_{1:n}}$ of $n$ i.i.d.  samples from the stage distribution $X_k$. 
The expected second order statistic $\expectation{X_{2:n}}$, i.e. the expected second highest value, from $n$ i.i.d. samples from distribution $X$ is equal to the expected revenue of a second price auction with $n$ buyers from $X$. For our Competition Complexity result we need a strict inequality between the expected second order statistic $\expectation{X_{2:n+c}}$ of $n+c$ samples and the expected first order statistic $\expectation{X_{1:n}}$ (the expected highest value, the social welfare) of $n$ samples. 
More specifically, we need to find a $c$ large enough for $\expectation{ X_{2:n+c} }$ to be at least as large as $\expectation{X_{1:n}}$. For MHR distributions strong tail bounds are known (see for example~\citet{cai2011extreme}), that can be used for approximations of the expected order statistics, by combining for example  with Markov-type inequalities\footnote{Incidentally, the approximations between the expected first and second order statistics we present here don't use this approach directly. We instead take an ``auction flavored'' approach and go through a comparison with the revenue of Myerson's auction.}. If we insist on strict bounds though, we cannot afford to use such lossy arguments, and new ideas are necessary. We postpone further discussion until the next Section where we formally define the model.
We proceed to explore the extent to which the assumption on the stage distributions can be relaxed.

\begin{inftheorem}
In the case of $n$ buyers and $\days$ independent stages, ex-post IR dynamic auctions, if $\days-1$ stage distributions have monotone hazard rate and the remaining stage distribution is regular, the bounds on the Competition Complexity and $\alpha$-approximate Competition Complexity are the same as Informal Theorem $1$.
\end{inftheorem}

This theorem requires much better upper bounds on the optimal dynamic revenue, which we obtain via an extension of the duality framework recently proposed by~\citet{cai2016}. The upper bound in~\citet{cai2016} improves on the trivial upper bound of the social welfare by substituting the value of each buyer's favorite item with the corresponding Myerson's virtual value. Our bound resembles this format: the largest expected first order statistic is substituted by the corresponding expected virtual value, i.e. optimal (static) revenue. Moreover, our dual solution also induces a virtual value function $\Phi_k(.)$ at each stage $k$.
The benchmarks we get from duality do not depend on ``tail assumptions'' like regularity and monotone hazard rate, but they do need stage distributions to be independent; the tail assumptions are necessary for the revenue of VCG to surpass the benchmark. Even though they can't be materialized into Competition Complexity bounds, we {\em can} prove some general upper bounds for correlated stages, under a stochastic dominance condition, using the same technique. We also use the duality framework as a tool in answering the following question: How does correlation affect the optimal revenue? In other words, moving from independent to correlated stages (while keeping the stage marginals the same) makes the optimal revenue increase or decrease? We defer details and further discussion until the next Section. 

Back to the Competition Complexity, can the tail assumptions on the stage distributions be relaxed even further? Surprisingly, the answer is no! If two stage distributions are regular, even if the stages are independent, there is only a single buyer and we ask for non-negative utility at each stage (ex-post IR), the $\alpha$-Competition Complexity is unbounded, for all constants $\alpha > 0$. We conclude this introduction with an example from~\citet{papadimitriou2016complexity} that shows that running the optimal static auction at each stage is not a constant approximation of the optimal dynamic auction; the lower bound on the $\alpha$-Competition Complexity is an immediate corollary.

\begin{example}[\citet{papadimitriou2016complexity}]\label{example: myerson every day}
Let $X_1$ and $X_2$ be the random variables indicating the value of the buyer for the first and second stage item. $X_1$ takes value $2^i$ with probability $2^{-i}$ for $i= 1,\dots,n$, and value $0$ with probability $2^{-n}$. $X_2$ takes value $2^i$ with probability $2^{-i}$ for $i=1,\dots,2^n$, and value $0$ with probability $2^{-2^n}$. It can be verified that the optimal static auction for both $X_1$ and $X_2$ extracts revenue at most $2$ ( Consider setting some price $2^k$. The expected revenue is at most $2^k \cdot \sum_{i \geq k} 2^{-i} \leq 2$ ). Therefore, running the optimal static auction at each stage extracts revenue at most $4$. On the other hand, there exists a dynamic mechanism that extracts revenue $n$: on the first stage the buyer pays her report $\hat{v}$. On the second stage the item is given for free with probability $\frac{\hat{v}}{\expectation{X_2}}$\footnote{Notice that $\expectation{X_2} > 2^n$, therefore $\frac{\hat{v}}{\expectation{X_2}}$ is a probability.}. An easy calculation shows that truthful reporting is a weakly dominant strategy. The revenue extracted is $\expectation{X_1} = n$.\footnote{A similar example can be constructed for continuous distributions. See~\citet{ashlagi2016}.} The intuition here is that if the expected value of the future item is large, the buyer is willing to pay her value on the first stage for a better probability of getting allocated the future item.
\end{example}

\section{Preliminaries and Main Results}

\subsection{Multi-Stage Auctions.} \label{subcsec: definitions for multi stage auctions}
A seller sells $\days$ items to $n$ buyers in $\days$ consecutive stages.
The value of buyer $i$ for the item on stage $k$ is $v^{i}_{k} \in V^i_k = \left[ \underline{v}^i_k, \overline{v}^i_k \right]$ and is distributed according to a random variable $X^i_k$. These random variables are independent among buyers, but for the same buyer can be correlated across stages, i.e. $X^i_k$ can be correlated with $X^{i}_{k'}$, but not with $X^{i'}_{k}$. $X^i_k$ has distribution $D^i_k$ with density $f^i_k$ and cumulative density $F^i_k$. 
It will be often convenient to use random variables rather than distributions and thus we use $X$ and $D$ interchangeably. Throughout this paper we use superscript to denote an agent and subscript to denote the stage. 
We write $\mathbf{X}_k = \prod_{i=1}^n X^i_k $ for the product distribution for stage $k$ (across all buyers).
Let $\mathcal{X}$ be the input to the seller's problem; $\mathcal{X}$ includes all the stage distributions, as well as their correlation. 
Since we're interested in Bulow-Klemperer type theorems, the distribution of $X^i_k$ is the same for all buyers $i$ and therefore we drop the superscript whenever applicable. 
We assume that the value for each item is revealed at the beginning of each stage: at the beginning of stage $k$, buyer $i$ knows her private history $v^i_{< k} = \left( v^i_{1}, v^i_{2}, \dots , v^i_{k-1} \right)$, her value $v^i_k$ for the item in stage $k$, as well as the public history $\mathbf{\hat{v}}_{<k}  = \left( \hat{v}^1_{1:k-1}, \hat{v}^2_{1:k-1}, \dots, \hat{v}^n_{1:k-1} \right)$, where $\hat{v}^i_{a:b} = \left( \hat{v}^i_a, \hat{v}^i_{a+1},\dots,\hat{v}^i_{b} \right)$ is the vector of reported values of agent $i$ for stages $a$ through $b$. 

From the revelation principle, it is sufficient to consider direct revelation mechanisms. A mechanism $M$ is a sequence of $\days$ allocation functions $\left( x_1, \dots x_{\days} \right)$ and $\days$ payment functions $\left( p_1, \dots p_{\days} \right)$, both taking as input all the reported valuations so far $\mathbf{v}_{\leq k} $. The allocation function for stage $k$ has a component $x^i_k \left( \mathbf{v}_{\leq k} \right)$ that represents the probability that buyer $i$ gets the item in stage $k$. Similarly, the payment function for stage $k$ has a component $p^i_k \left( \mathbf{v}_{\leq k} \right)$ for the payment of buyer $i$ in stage $k$.  
A mechanism is feasible if for all stages $k$ and all histories $\mathbf{v}_{\leq k} $, $x^i_k \left( \mathbf{v}_{\leq k} \right) \in [0,1]$ for all agents $i$, and $\sum_{i=1}^n x^i_k \left( \mathbf{v}_{\leq k} \right) \leq 1$.
We assume quasi-linear utilities; the utility of buyer $i$ in stage $k$ is $v^i_k \cdot x^i_k \left( \mathbf{v}_{\leq k} \right) - p^i_k\left( \mathbf{v}_{\leq k} \right)$. 


\paragraph{Incentive Compatibility.} 
At stage $k$ we would like to have a mechanism where agent $i$, with real value $v^i_k$, maximizes her utility when reporting $v^i_k$, among all possible reports $\hat{v}^i_k$. This utility is in expectation over the other agents' current values, as well as her own and other agents' future values. When deciding what value $\hat{v}^i_k$ to report in stage $k$, the agent has to take into account that the future allocation and payments (and therefore the future utility) will be affected by this report. It could be the case that lying only in stage $k$ or lying only in stage $k+1$ results in lower overall utility, but lying in both stages results in a higher utility!  Thus, when deciding when to lie, the buyer must choose in advance a strategy that deviates from the truth now and in the future.

Let $\mathcal{S}_k$ be the set of all ``future deviation strategies'' at stage $k$. 
A deviating strategy $s \in \mathcal{S}_k$ is a function from possible ``futures'', i.e. elements of $\prod_{i=1}^n \prod_{t=k+1}^{\days} V^i_t$ to possible reports, i.e. elements of
$\left( \hat{v}_{k+1}, \dots, \hat{v}_{\days} \right) \in V_{k+1:\days} = \prod_{i=k+1}^{\days} V_i$. 
Note that in our definition of $\mathcal{S}_k$ the deviation in stage $k$ is not included.
A mechanism is incentive compatible if, for every buyer $i$, every stage $k$, all possible histories $\mathbf{v}_{< k}$, for all values $v^i_k \in V_k$ on stage $k$, and all possible current deviations $\hat{v}^i_k \in V_k$ and future deviation strategies $s \in \mathcal{S}_k$:
\begin{multline}\label{eq: IC constraint}\textstyle
\mathbb{E}_{ \mathbf{v}^{-i}_{\geq k}, v^i_{k+1:\days} } \left[ \sum_{j\geq k} v^i_j x^i_j\left( \mathbf{v}^{-i}_{\leq j}, v^i_{1:k-1}, v^i_k , v^i_{k+1:j} \right) - p^i_j\left( \mathbf{v}^{-i}_{\leq j}, v^i_{1:k-1}, v^i_k , v^i_{k+1:j} \right) \right] \geq \\ \textstyle
\mathbb{E}_{ \mathbf{v}^{-i}_{\geq k}, v^i_{k+1:\days} } \left[ \sum_{j\geq k} v^i_j x^i_j\left( \mathbf{v}^{-i}_{\leq j}, v^i_{1:k-1} \hat{v}^i_k, s\left( \mathbf{v}^{-i}_{k:j}, v^i_{k:j} \right) \right) - p^i_j\left( \mathbf{v}^{-i}_{\leq j},v^i_{1:k-1}, \hat{v}^i_k, s\left( \mathbf{v}^{-i}_{k:j}, v^i_{k:j} \right) \right)  \right]
\end{multline}

Intuitively, buyer $i$ at stage $k$ compares her expected utility when telling the truth now and in the future, with her expected utility for reporting $\hat{v}^i_k$ now, $s\left( \mathbf{v}^{-i}_{k:j}, v^i_{k:j} \right)$ in stage $j$, where $\mathbf{v}^{-i}_{k:j}$ is the rest of the buyers' values in stages $k$ through $j$ and $v^i_{k:j}$ are the true values of buyer $i$ in stages $k$ through $j$. 
If $\mathcal{S}_k$ is the set of all function from $\prod_{i=1}^n \prod_{t=k+1}^{\days} V^i_t$ to $\prod_{i=k+1}^{\days} V_i$, and Equation~\ref{eq: IC constraint} is satisfied, we say that the mechanism is \textit{incentive compatible in a perfect Bayesian equilibrium}. If $\mathcal{S}_k$ only includes the ``identity function'', i.e. the agent assumes truthful reporting for future stages, we say that the mechanism is \textit{periodic incentive compatible}.


\paragraph{Individual Rationality.} For IR we focus on the two extremes: ex-post and ex-ante. The latter notion asks for every buyer's expected utility at every stage to be non-negative. Formally, for every buyer $i$, every stage $k$, all possible histories $\mathbf{v}_{< k}$, for all values $v^i_k \in V_k$ on stage $k$:

\begin{equation}\label{eq: ex-ante IR}
\textstyle \mathbb{E}_{ \mathbf{v}^{-i}_{\geq k}, v^i_{k+1:\days} } \left[ \sum_{j\geq k} v^i_j x^i_j\left( \mathbf{v}^{-i}_{\leq j}, v^i_{k:j} \right) - p^i_j\left( \mathbf{v}^{-i}_{\leq j}, v^i_{k:j} \right) \right] \geq 0
\end{equation}

Ex-post individual rationality asks for every buyer's utility to be non-negative \textit{at every stage}, no matter what the other buyers' valuations are. Formally, for every buyer $i$, stage $k$, and possible history $\mathbf{v}_{\leq k}$:
\begin{equation}\label{eq: ex-post IR}
\textstyle v^i_k \cdot x^i_k \left( \mathbf{v}_{\leq k} \right) - p^i_k\left( \mathbf{v}_{\leq k} \right) \geq 0.
\end{equation}

\paragraph{The seller's problem.} The seller's goal is to find the revenue optimal mechanism that is incentive compatible and individually rational. Let $\opt{\mathcal{X},n,\days}$ denote the revenue of the optimal mechanism for $n$ buyers and $\days$ stages, when the buyers' valuations are drawn according to $\mathcal{X}$.
For the special case of $\days = 1$ the solution is given by~\citet{myerson1981optimal}. For a general $\days$, the seller's problem can be expressed as a linear program with variables $x^i_k$ and $p^i_k$, objective
\[\textstyle \max \mathbb{E} \left[ \sum_{k=1}^{\days} \sum_{i=1}^n p^i_k \left( \mathbf{v}_{\leq k} \right) \right], \]
subject to constraints~\ref{eq: IC constraint}, constraints~\ref{eq: ex-ante IR} or~\ref{eq: ex-post IR}, as well as the allocation $x$ being feasible (a linear constraint). Note that the seller's revenue is smaller for ex-post IR than ex-ante IR. Furthermore, it weakly decreases as the set $\mathcal{S}_k$ of deviations considered becomes larger. Therefore, the best upper bounds possible would be for ex-ante IR and periodic IC.

\subsection{Competition Complexity.}

Let $\rev{M,\mathcal{X},n,\days}$ denote the expected revenue of running auction $M$ at every stage, for $\days$ stages, with $n$ buyers whose values are drawn according to $\mathcal{X}$. We are interested in the number $c$ of extra buyers necessary such that $\rev{VCG,\mathcal{X},n+c,\days}$ is at least $\opt{\mathcal{X},n,\days}$, where $VCG$ is simply a second price auction. This number is called the \textit{Competition Complexity} with respect to $VCG$, defined by~\citet{eden2016competition}. We also study \textit{approximations}: 

\begin{definition}
The $\alpha$-approximate Competition Complexity with respect to $VCG$ is the minimum number $c$ such that $\rev{VCG,\mathcal{X},n+c,\days}$ is at least $\alpha \cdot \opt{\mathcal{X},n,\days}$.
\end{definition}

Note that we can define the $\alpha$-approximate Competition Complexity and Competition Complexity to be with respect to any prior-independent auction $M$. In this work we focus on $VCG$. 
At a high level our approach is the following: (1) Find an upper bound $B\left( \mathcal{X},n,\days \right)$ to $\opt{\mathcal{X},n,\days}$. (2) Prove that the revenue of running a second price auction at every stage with $c$ additional buyers (present in every stage) yields revenue at least $B\left( \mathcal{X},n,\days \right)$.
We present each step separately.
For different distributions and different constraints the bounds in steps (1) and (2) are different.
Our final bounds on the Competition Complexity come from mixing and matching these different bounds.

\subsection{Upper bounds on $\opt{\mathcal{X},n,m}$.}


Buyers' valuations $X^i_k$ on stage $k$ are independent draws from a distribution $X_k$. We can arrange the values in a descending order: $\left( X_{k} \right)_{1:n} \geq \left( X_{k} \right)_{2:n} \geq \dots \geq \left( X_{k} \right)_{n:n}$. We call $\left( X_{k} \right)_{r:n}$ the $r$-th order statistic\footnote{We write $X_{r:n}$ for the $r$-th order statistic of $n$ samples from a distribution $X$. In the auctions context we write $v_{t:t'} = \left( v_t, v_{t+1}, \dots, v_{t'}\right)$ for the reported values in stages $t$ through $t'$. It will be clear from context which of the two notions we refer to.}. Given a product distribution $\mathbf{X}_k = \prod_{i=1}^n X_k $, let $\myerson{\mathbf{X}_k}$ be the revenue of Myerson's optimal auction, i.e. the revenue optimal mechanism for a one-shot auction where buyers' valuations are drawn i.i.d. from $X_k$. Our upper bounds on $\opt{\mathcal{X},n,m}$ do not require any tail assumptions.

\paragraph{Social welfare.}
Our first upper bound on $\opt{\mathcal{X},n,m}$ is the trivial one: the social welfare. At every stage $k$ we can extract revenue at most the expected maximum valuation at that stage. 
\begin{claim}\label{claim: trivial bound} 
$ \opt{\mathcal{X},n,m} \leq \sum_{k=1}^{\days} \expectation{ \left( X_{k} \right)_{1:n} },$
where $\left( X_{k} \right)_{1:n}$ is the highest-order statistic of $n$ i.i.d. samples drawn from $X_k$.
\end{claim}
Note that this bound holds even for ex-ante IR and periodic IC mechanisms, and even if the stages are correlated\footnote{Recall that we allow the value $v^i_k$ of agent $i$ in stage $k$ to be correlated with her value in stage $k'$. We don't allow this value to be correlated with some other value $v^j_k$ of a different agent.}. Surprisingly, as we see later, if we restrict the marginal distributions at each stage to have monotone hazard rate we can still provide bounds on the Competition Complexity even with this trivial upper bound.

\paragraph{Duality based bounds.}
To improve on the trivial bound we ask for (1) independent stages, (2) ex-post IR mechanisms. Ex-post IR dynamic mechanisms have received a lot of attention in recent works, e.g.~\citet{papadimitriou2016complexity,ashlagi2016,mirrokni2016dynamic,mirrokni2016oblivious,balseiro2016dynamic}, because of their practicality\footnote{Ex-ante IR mechanisms are less appealing because of the following unsettling property: the seller can extract revenue equal to the social welfare by (roughly) charging the social welfare upfront, before the buyers have realized their future values. This makes sense for utility maximizers in a stylized model, but limits the applicability of the auctions produced.}.
Our improved bound is as follows. Choose any stage $j$; the contribution of stage $j$ is $\myerson{ \mathbf{X}_j }$, the optimal revenue we would extract from stage $j$ in a non-dynamic setting. The contribution of every other stage $k\neq j$ is the expected maximum $\expectation{ \left( X_{k} \right)_{1:n} }$ of $n$ samples drawn from $X_k$.

\begin{lemma}\label{lem:upper bound on opt}
For independent stages, ex-post IR and periodic IC dynamic mechanisms
\[ \textstyle \opt{\mathcal{X},n,m} \leq \min_{j=1,\dots,\days} \left\lbrace \myerson{ \mathbf{X}_j } + \sum_{k=1, k \neq j}^{\days} \expectation{ \left( X_{k} \right)_{1:n} } \right\rbrace, \]
where $\left( X_{k} \right)_{1:n}$ is the highest-order statistic of $n$ samples drawn from $X_k$, and $\mathbf{X}_j$ is the product distribution for stage $j$.
\end{lemma}


Our proof is via an extension of the Cai-Devanur-Weinberg duality framework to the dynamic setting. 
The work of~\citet{cai2016} unified many different recent advances in Bayesian mechanism design by providing an approximately tight upper bound for the optimal revenue using a single dual solution. In their work, they start from a certain linear program for revenue maximization and Lagrangify the Bayesian IC and IR constraints. 
By Lagrangifying the IC and IR constraints of our linear program for maximizing revenue in the dynamic setting we can get similar characterizations. 
The bound of~\citet{cai2016} improves on the social welfare upper bound by substituting the value of each buyer's favorite item with the corresponding Myerson's virtual value. Our bound substitutes the largest expected order statistic with the corresponding expected virtual value. 
In order to develop some intuition we first prove the single agent, two stage case in Section~\ref{sec: single agent upper bound}. The general proof can be found in Appendix~\ref{sec: proof of upper bound}.

One can prove Lemma~\ref{lem:upper bound on opt} by using a characterization Lemma of~\citet{ashlagi2016} for the optimal dynamic, ex-post IR mechanism in this setting. Regardless, we believe that the extension of Cai-Devanur-Weinberg framework in the dynamic setting is an important conceptual contribution of this paper and we expect it to have further applications in the design and analysis of dynamic auctions in more general settings. To this end, even though it can't be materialized into a Competition Complexity result, we use the duality framework to prove a bound similar to Lemma~\ref{lem:upper bound on opt} for correlated stages that satisfy a stochastic dominance structure; details can be found in Appendix~\ref{app: correlated upper bound}.
We also use the framework in the revenue non-monotonicity Lemma presented next.

\paragraph{Discrete vs Continuous.} For simplicity of presentation, we prove Lemma~\ref{lem:upper bound on opt} for distributions with discrete support. Our proof can be easily modified to hold for continuous distributions.

\subsubsection*{Who benefits from correlation? A revenue non-monotonicity lemma.}

Given Lemma~\ref{lem:upper bound on opt}, a bound for the revenue of dynamic auctions with independent stages, one can ask if the same bound holds for correlated stages. A perhaps more natural question is: \textit{Does correlation help the buyer or the seller?} For example, in the case of two stages and a single agent, let $X_1$ and $X_2$ be the random variables denoting the buyer's valuations for the first and second stage, and let $\textsc{OPT}$ be expected revenue of the optimal IC and ex-post IR auction, in the case that $X_1$ and $X_2$ are independent. Let $\overline{\textsc{OPT}}$ be expected revenue of the optimal auction if $X_1$ and $X_2$ are correlated (but the marginal distributions for each stage remain the same). Is it the case that $\textsc{OPT}$ is always larger or always smaller than $\overline{\textsc{OPT}}$? 
There are two opposing intuitions. On one hand, one could argue that when the buyer reports her value on the first stage she reveals more information to the seller if the stages are correlated compared to when the stages are independent.
Therefore, the seller has to pay ``information rents'' to the buyer, and thus the expected revenue of the correlated stages is lower than the expected revenue of independent stages with identical marginals. 
On the other hand, before the buyer reports her first stage value, the seller knows more about how the buyer's values evolve across stages. This additional knowledge could be used to increase revenue compared to the independent stages case.
In Section~\ref{sec:non monotonicity} we show that both scenarios are possible. 
\begin{lemma}
Moving from independent to correlated stages while keeping the marginal (stage) distributions identical may increase or decrease the revenue of the optimal dynamic auction.
\end{lemma}

Proving such a statement requires being able to identify the optimal dynamic auction, or a good upper bound for the revenue of the optimal dynamic auction, for independent as well as correlated stages. Even for two fixed distributions $X_1$ and $X_2$ it is not clear how to argue about what the optimal dynamic auction is capable of when $X_1$ and $X_2$ are correlated. Furthermore, $\expectation{X_1} + \expectation{X_2}$ will not do; the bound must be affected by the correlation. We provide upper bounds via duality; we cite this as another usage of the duality framework developed in this paper.

\subsection{Lower bounds on the revenue of Vickrey.}

After finding suitable bounds for $\opt{\mathcal{X},n,m}$, we need to show that the revenue of VCG (a second price auction at each stage) with additional buyers surpasses these bounds. A first observation is that the expected revenue of VCG is the expected second order statistic:

\begin{obs}
The expected revenue of a second price auction with $n$ agents whose values are drawn i.i.d. from $X$ is $\expectation{ X_{2:n} }$.
\end{obs}

Our upper bounds on $\opt{\mathcal{X},n,m}$ ( Claim~\ref{claim: trivial bound} and Lemma~\ref{lem:upper bound on opt} ) are sums over $\days$ terms. The term that corresponds to stage $k$ is either the expected revenue of the optimal static auction $\myerson{\mathbf{X}_k}$ for that stage, or the expected highest order statistic $\expectation{X_{1:n}}$ of $n$ samples. We upper bound each term separately. This gives us a sufficient number of extra buyers, i.e. an upper bound on the Competition Complexity. For the terms involving the optimal static auction, the original Theorem of Bulow and Klemperer provides a good bound for $\expectation{X_{2:n+c}}$ and regular distributions:

\begin{theorem}[Bulow-Klemperer]\label{thm: bulowklemperer}
Let $X$ be a random variable with a regular distribution $D$. Let $\mathbf{X}$ be the product distribution of $n$ samples drawn i.i.d from $D$. Then $\expectation{ X_{2:n+1} } \geq \myerson{ \mathbf{X} }$.
\end{theorem}

The following corollary can be shown:
\begin{corollary}\label{cor: second price is n-1/n approximation}
Let $X$ be a random variable with a regular distribution $D$. Let $\mathbf{X}$ be the product distribution of $n$ samples drawn i.i.d from $D$. Then $\expectation{ X_{2:n} } \geq \frac{n-1}{n} \myerson{ \mathbf{X} }$. In other words, the revenue of a second price auction is an $\frac{n-1}{n}$ approximation to the revenue of the optimal auction.
\end{corollary}

We first note that for terms involving the expected highest order statistic similar bounds are impossible for regular distributions, as we've already seen in Example~\ref{example: myerson every day}:

\begin{example}
Let $X$ be a random variable following the equal revenue distribution with $F(x) = 1 - \frac{1}{x}$, for $x \in [1,\infty)$. $\expectation{X} = \int_0^{\infty} \frac{1}{x} dx$ is unbounded, while $\expectation{X_{2:n}} = n-1$, i.e. bounded for all $n$.\footnote{See Claim~\ref{claim: equal revenue order statistics} in Appendix~\ref{sec: lower bounds on cc} for a calculation.}.The example can be modified to hold for truncated distributions: for all $n$, there exists a truncation value $V$, such that for the truncated distribution $\expectation{ X } > \expectation{X_{2:n}}$.
\end{example}

Therefore, in order to get a bound on the Competition Complexity we need to impose a restriction stronger than regularity on some stage distributions. A natural candidate is distributions with \textit{Monotone Hazard Rate}; a distribution has \textit{Monotone Hazard Rate} (MHR) if its hazard rate $h(v) = \frac{f(v)}{1-F(v)}$ is monotone non-decreasing. MHR distributions are a subset of regular distributions and have various nice properties, like bounded expected order statistics, small tails, etc.
In this paper we show the following new bounds:
\begin{theorem}\label{thm: bounds on mhr stuff}
Let $X_{r:n}$ be the $r$-th order statistic of $n$ i.i.d samples from a continuous distribution $X$ with monotone hazard rate. Then:
(1) $\expectation{X_{2:4n}} \geq \expectation{X_{1:n}}$.
(2) $\expectation{X_{2:n+1}} \geq \frac{1}{e} \expectation{X_{1:n}}$; (3) $\expectation{X_{2:n}} \geq \frac{1}{3} \expectation{X_{1:n}}$.
\end{theorem} 
{}
MHR distributions have been studied extensively in the Statistics literature under the (perhaps better) name of IHR, Increasing Hazard Rate, and IFR, Increasing Failure Rate (e.g.~\citet{barlow1966inequalities},~\citet{barlow1996mathematical})\footnote{To keep consistency with the auctions community we refer to them as MHR distributions in this paper.}.
A common trick when working with MHR distributions is to write the cdf as $F(x) = 1 - e^{-\int_0^{x} h(z) dz }$\footnote{See Claim~\ref{claim: basic MHR } in Appendix~\ref{appendix:missing from mhr} for an explanation.}. Then, since $h(x)$ is non-decreasing, $H(x) = \int_0^{x} h(z) dz$ is a convex function. Using one sided bounds for $H(x)$ (for example a linear approximation) one can provide lower bounds and upper bounds for quantities like the expected minimum of two samples (e.g. \cite{blog, dhangwatnotai2015revenue, hartline2009simple}). When working with order statistics of many samples though, one sided bounds such as these do not work, since the closed form for the expected second order statistic has both positive and negative terms that involve $H(x)$. Moreover, taking more samples will not compensate for a lossy argument; the proofs need to work for distributions that are essentially point masses, where all the order statistics are equal for any number of samples.

For each of the bounds in Theorem~\ref{thm: bounds on mhr stuff} we take a different approach. 
For the first one, we start by showing a one sample bound: $\expectation{X_{2:4}} \geq \expectation{X}$. This bound is tight for the exponential distribution (in a sense that $\expectation{X_{2:3}} < \expectation{X}$). We first prove that the inequality is strict when $H(x)$ is piece-wise linear and convex, and then show that every convex function with bounded domain can be approximated by a piece-wise linear function. We combine the one sample result, the fact that order statistics of MHR distributions have MHR distributions themselves, and a coupling argument to generalize to $n$ samples.
For the second bound, $\expectation{X_{2:n+1}} \geq \frac{1}{e} \expectation{X_{1:n}}$, we take an ``auction flavored'' approach. First, the LHS is at least $\myerson{ \mathbf{X} }$ using Bulow and Klemperer's result. 
Second, we compare $\myerson{ \mathbf{X} }$ with $\myerson{ X_{1:n} }$, the expected revenue of the optimal auction in a (one-shot) single agent auction with distribution $X_{1:n}$, using a coupling argument. Third, since order statistics of MHR distributions have MHR distributions, we can use known bounds to compare $\myerson{ X_{1:n} }$ and $\expectation{ X_{1:n} }$.
For the last bound, $\expectation{X_{2:n}} \geq \frac{1}{3} \expectation{X_{1:n}}$, we combine a (known) bound on the expected minimum of two samples from an MHR distribution with the (also known) fact that spacings of order statistics of MHR distributions, i.e. $\expectation{X_{1:n}} - \expectation{X_{2:n}}$, are non-increasing functions of the number of samples.
We prove Theorem~\ref{thm: bounds on mhr stuff} as three separate Lemmas in Section~\ref{sec:second price bounds}.

\paragraph{Discrete vs Continuous.} 
For a distribution over a discrete domain $\{ 1, 2, \dots, N \}$, the definition of hazard rate is $h(i) = \frac{p(i)}{\sum_{j \geq i} p(j)}$ (see~\citet{barlow1996mathematical}). Some known inequalities for MHR distributions fail for the discrete case. For example, 
for continuous MHR distributions one can show that $\Pr[ X \geq \expectation{X} ] \geq \frac{1}{e}$; the inequality fails for a geometric distribution. Our proofs hold only for continuous distributions. It remains open whether the statements are true for discrete MHR distributions.

\subsection{Putting everything together.}

By combining the different upper bounds on $\opt{\mathcal{X},n,m}$ with the corresponding lower bounds for VCG we can get upper bounds for the Competition Complexity and approximate Competition Complexity of dynamic auctions. We prove our lower bounds for the Competition Complexity in Appendix~\ref{sec: lower bounds on cc}. Our lower bounds work for (1) independent stages, $\days$ MHR distributions, for ex-ante IR auctions (applied in Theorem~\ref{thm: main thm all mhr}), and (2)  independent stages, $\days-1$ MHR and $1$ regular distribution, for ex-post IR auctions (applied in Theorem~\ref{thm: main thm almost all mhr}). 
The proofs are similar. The auction in the second bound is a generalization of Example~\ref{example: myerson every day}. The auction in the first bound exploits an unsettling feature of ex-ante IR mechanisms: the buyers are willing to give up their expected future (net) utility just to be able to participate in the future auction. An ex-ante IR auction that extracts all of the social welfare is the following: at every stage the seller runs a second price auction, but before that, all buyers pay an entree fee equal to their expected utility for participating (expected value subject to being the winner, minus expected second highest value, multiplied by probability of winning). A common difficulty in both proofs is the algebraic manipulations of the expected first and second order statistics\footnote{The Lambert W-function makes an appearance.}. 
Combining with the upper bounds we get the following Theorems:

\begin{theorem}\label{thm: main thm all mhr}
For a dynamic environment where every stage distribution $X_k$ is continuous, and has monotone hazard rate then, even for periodic IC and ex-ante IR dynamic auctions (by combining Claim~\ref{claim: trivial bound} + Theorem~\ref{thm: bounds on mhr stuff}):
(1) The Competition Complexity is at least $(e-1)n$ and at most $3n$; (2) The $\frac{1}{e}$-approximate Competition Complexity is $1$; (3) For $n \geq 2$, the $\frac{1}{3}$-approximate Competition Complexity is $0$.
\end{theorem}

\begin{theorem}\label{thm: main thm almost all mhr}
For a dynamic environment where every stage distribution $X_k$ is continuous, the stage distributions are independent, $\days-1$ stage distributions have monotone hazard rate and the remaining stage distribution is regular, then even for periodic IC and ex-post IR dynamic auctions: (1) The Competition Complexity is at least $(e-1)n$ and at most $3n$ (For the upper bound: Lemma~\ref{lem:upper bound on opt} + Thms~\ref{thm: bulowklemperer},~\ref{thm: bounds on mhr stuff}); (2) The $\frac{1}{e}$-approximate Competition Complexity is $1$ (Lemma~\ref{lem:upper bound on opt} + Thms~\ref{thm: bulowklemperer},~\ref{thm: bounds on mhr stuff}); (3) For $n \geq 2$, the $\frac{1}{3}$-approximate Competition Complexity is $0$. (Lemma~\ref{lem:upper bound on opt} + Corollary~\ref{cor: second price is n-1/n approximation}, Thm~\ref{thm: bounds on mhr stuff})
\end{theorem}

Our lower bound for the Competition Complexity in Theorem~\ref{thm: main thm almost all mhr} uses $\days-1$ MHR distributions and one regular distribution. It remains open whether the Competition Complexity is sublinear for the case of $\days$ independent and MHR stages.


\subsection{Related Work}

\paragraph{Dynamic Auctions.}

Dynamic mechanisms have been studied extensively in quite general settings; see \citet{bergemann2011dynamic} for a recent survey. Many works study problems where agents arrive and depart dynamically, e.g. \citet{parkes2004mdp, pai2008optimal,gershkov2009dynamic,gershkov2010efficient}, or problems with evolving private information, e.g. \citet{courty2000sequential, kakade2013optimal,pavan2014dynamic,cavallo2008efficiency,cavallo2012optimal,krahmer2015optimal}.

The setting studied in this paper, where ex-post individual rationality is a hard constraint, was first studied in~\citet{papadimitriou2016complexity}. The authors prove that computing the optimal deterministic dynamic mechanism is NP-hard even for two stages and a single agent, as well as very strong separations between adaptive and non-adaptive dynamic mechanisms. They also show that the optimal randomized mechanism can be computed via an LP whose size is polynomial in number of types and in the number of days. In this paper we are not interested in computation and therefore use a much simpler (and much larger) Linear Program.

\citet{ashlagi2016} provide characterizations of the optimal ex-post IR, periodic IC dynamic mechanism, with $\days$ \textit{independent} stages and $n$ buyers. They show that there exists an optimal mechanism that has stage utility equal to zero for all stages, except maybe the last. As implied by our upper bound in Lemma~\ref{lem:upper bound on opt}, in last stage the seller might have to pay the buyers\footnote{To see this most clearly, consider a single agent, two stage situation where $X_1$ and $X_2$ are such that $\expectation{X_1} > \myerson{X_1} + \expectation{X_2}$; the RHS is an upper bound to the optimal revenue by Lemma~\ref{lem:upper bound on opt} in this paper. The characterization of~\citet{ashlagi2016} says that there exists an optimal mechanism that extracts $\expectation{X_1}$ in the first stage; in the second stage the seller must pay back at least $\expectation{X_1} - \left( \myerson{X_1} + \expectation{X_2} \right)$}. Surprisingly, their mechanism can be described via updates, at every stage, to a scalar variable that guides the future allocation and payments. The authors use this characterization to give a mechanism that obtains a $\frac{1}{2}$ approximation to the optimal revenue for the single buyer problem. 


\citet{mirrokni2016dynamic} study dynamic mechanisms with an interim IR constraint. They define a class of mechanisms called \textit{bank account} mechanisms. Bank account mechanisms maintain a state variable, the balance, that is updated throughout the execution of the mechanism depending on a ``spending'' and ``depositing'' policy. The allocation and payment at each stage depend on the report and the balance.~\citet{mirrokni2016optimal} study revenue maximization for bank account mechanisms subject to an ex-post IR constraint. 
Closer to this work,~\citet{mirrokni2016oblivious} study the design of \textit{oblivious} dynamic mechanisms. An oblivious dynamic mechanism decides on the allocation and payment for stage $k$ using information only about the current and past stages, i.e. it is oblivious about the buyers' value distributions $D_{k+1},\dots,D_m$. Their mechanism \textit{ObliviousBalance} runs at each stage a combination of Myerson's optimal auction, a second price auction, and the money burning mechanism of~\citet{hartline2008optimal}. Their mechanism obtains a $\frac{1}{5}$ approximation to the optimal revenue. Here we can get better approximations with much simpler mechanisms, without assuming knowledge of distributions $X_1,\dots,X_k$ at the beginning of stage $k$, at the small cost of tail assumptions like monotone hazard rate.

\paragraph{Bulow-Klemperer Type Results.}

Prior-independent mechanisms have been developed in both single and multi-dimensional \textit{static} settings, e.g. \cite{azar2013optimal, devanur2011prior, goldner2016prior, segal2003optimal}. \citet{sivan2013vickrey} give a version of the Bulow-Klemperer theorem for non-i.i.d. irregular distributions. For the case of multi-unit bidders,~\citet{roughgarden2015robust} provide a   Bulow-Klemperer type result for multidimensional static auctions, where the benchmark is the optimum dominant strategy incentive compatible (DSIC) mechanism. 
Closer to our work,~\citet{eden2016competition} provide a Bulow-Klemperer result for multidimensional static auctions, i.e. without any loss or approximation, against the stronger benchmark of the optimum Bayesian incentive compatible and Bayesian individual rational mechanism. They introduce the term \textit{Competition Complexity}, that we also adopt here. 
Their main result is that the Competition Complexity of $n$ buyers with additive valuations over $m$ independent, regular items is at most $n + 2m - 2$ and at least $\log(m)$. Their upper bounds on the optimal static revenue is also via an extension of the duality framework of~\citet{cai2016}. More recently,~\citet{rubinstein2017} also study a relaxed notion of Competition Complexity; they show that when auctioning $m$ items separately the 99\%-Competition Complexity is $O(\log m)$, and (for regular distributions) this further goes down to constant when auctioning the items as one bundle. A closely related line of work considers mechanism design with limited information in the form of samples, e.g. \cite{dhangwatnotai2015revenue,cole2014sample,devanur2016sample,morgenstern2015pseudo}.

\paragraph{The Cai-Devanur-Weinberg Duality Framework, Extensions and Related Techniques.}

Multiple strong duality frameworks have been developed recently, e.g. \cite{daskalakis2013mechanism, daskalakis2015strong, giannakopoulos2015selling, giannakopoulos2014duality}, that can be seen as an optimal transport/bipartite
matching problem. \citet{haghpanah2015reverse} provide an alternative strong duality framework. Closer to the framework we extend,~\citet{carroll2017robustness} takes a partial Lagrangian over IC and IR constraints; the application is a screening problem. As already discussed, in this paper we present an extension of the duality framework of~\citet{cai2016} for dynamic settings. This framework was used to unify  and improve the results of several recent works on Bayesian mechanism design (e.g. \cite{hart2012approximate, li2013revenue, babaioff2014simple, yao2015n, chawla2007algorithmic, chawla2010multi,chawla2015power} ).
It was recently extended by~\citet{cai2016simple} to prove approximation results for simple mechanisms in settings with multiple subadditive bidders, and by~\citet{brustle2017approximating} for a two-sided market setting. It was also extended in a different way by~\citet{eden2016simple} for a single buyer with values that exhibit a ``limited complementarity'' property.
\section{Warm up: One Buyer, Two Independent Stages}\label{sec: single agent upper bound}

In this Section we prove the special case of Lemma~\ref{lem:upper bound on opt} for one buyer and two stages.

\begin{lemma}\label{lem:upper bound on opt one agent}
For single agent, two independent stages, ex-post IR and periodic IC dynamic mechanisms
\[ \opt{\mathcal{X},1,2} \leq \min \left\lbrace \myerson{ X_1 } + \expectation{  X_2 },  \expectation{  X_1 } + \myerson{ X_2 } \right\rbrace. \]
\end{lemma}

\subsection{The Partial Lagrangian.}

The optimal dynamic auction needs to satisfy the following two types of constraints:

\begin{itemize}
\item Periodic incentive compatibility (PIC). At any stage $k$, assuming truthfulness in the future stages, truthfully revealing $v_k$ maximizes the buyer's expected utility, among all possible reports $\hat{v}_k$. 
For the first stage, this constraint can be expressed as: for all $v_1, \hat{v}_1$ in $V_1$
\begin{multline*}
v_1 x_1( v_1 ) - p_1(v_1) + \mathbb{E}_{v_2 \in V_2} \left[ v_2 x_2\left( v_1, v_2 \right) - p_2\left( v_1, v_2 \right)\right] \geq  \\
v_1 x_1( \hat{v}_1 ) - p_1(\hat{v}_1) + \mathbb{E}_{v_2 \in V_2} \left[ v_2 x_2\left( \hat{v}_1, v_2 \right) - p_2\left( \hat{v}_1, v_2 \right)\right].
\end{multline*}

For the second stage: for all $v_1$ in $V_1$, and all $v_2, \hat{v}_2$ in $V_2$
\begin{equation*}
v_2 x_2( v_1, v_2 ) - p_2( v_1, v_2 ) \geq v_2 x_2( v_1, \hat{v}_2 ) - p_2( v_1, \hat{v}_2 ).
\end{equation*}

\item Ex-post individual rationality. The buyer's stage utility is non-negative at every stage $k$, no matter what the reports were in the previous stage (in the case of stage $2$).
\begin{equation*}
v_k x_k( v_{\leq k} ) - p_k(v_{\leq k} ) \geq 0
\end{equation*}

\end{itemize}

The revenue objective can be written as:

\[\mathbb{E}_{v_1,v_2} \left[ p_1(v_1) + p_2(v_1,v_2) \right]\]

Thus, we have the following primal program\footnote{The support is discrete for simplicity of presentation.}:

  \begin{align*}
  \textrm{max }  &   \sum_{v_1 \in V_1}f(v_1)p_1(v_1) + \sum_{v_1 \in V_1} f(v_1)\sum_{v_2 \in V_2}f(v_2)p_2(v_1, v_2)  \notag\\
  \textrm{subject to: } \notag \\
   \forall v_1,\hat{v}_1 \in V_1: & \qquad v_1x_1(v_1) - p_1(v_1) + \sum_{v_2 \in V_2}f(v_2) \left( v_2x_2(v_1, v_2) - p_2(v_1,v_2) \right) \geq \\
  & \qquad \qquad \qquad v_1x_1(\hat{v}_1) - p_1(\hat{v}_1) + \sum_{v_2 \in V_2}f(v_2) \left( v_2x_2(\hat{v}_1, v_2) - p_2(\hat{v}_1,v_2) \right)  \\ 
   \forall v_1\in V_1,\forall v_2,\hat{v}_2\in V_2: & \qquad v_2 x_2( v_1, v_2 ) - p_2( v_1, v_2 ) \geq v_2x_2(v_1, \hat{v}_2) - p_2(v_1, \hat{v}_2)  \\
   \forall v_1\in V_1: & \qquad v_1 x_1( v_1 ) - p_1(v_1 ) \geq 0 \\
   \forall v_1\in V_1,v_2\in V_2: & \qquad v_2 x_2( v_1,v_2 ) - p_2(v_1,v_2 ) \geq 0 \\
   \forall v_1\in V_1: & \qquad x_1( v_1 ) \in [0,1] \\
   \forall v_1\in V_1, \forall v_2 \in V_2: & \qquad x_2( v_1, v_2 ) \in [0,1]
  \end{align*}

We introduce a dual variable $\lambda_k(v_{\leq k},\hat{v}_k)$ for the  periodic IC constraints for stage $k$ and a dual variable $\kappa_k(v_{\leq k})$ for the ex-post IR constraints for stage $k$. In other words, the dual variables are $\lambda_1(v_1,\hat{v}_1)$, $\lambda_2(v_1,v_2,\hat{v}_2)$, $\kappa_1(v_1)$ and $\kappa_2(v_1,v_2)$. \citet{cai2016} include the IR constraints with the IC constraints, by introducing a null type $\perp$, with zero allocation and zero payment; in our case, this is possible only for the ex-post IR constraint in the last stage. Similarly, we do not take Lagrangian multipliers for the feasibility constraints.
The partial Lagrangian $\mathcal{L}(\lambda, \kappa, x, p)$  of the primal program is as follows:

\begin{align*}
\mathcal{L}(\lambda, \kappa, x, p) =& \sum_{v_1 \in V_1} f(v_1) p_1(v_1) + \sum_{v_1 \in V_1}\sum_{v_2 \in V_2} f(v_1) f(v_2) p_2(v_1,v_2) \\
& + \sum_{v_1 \in V_1}\sum_{\hat{v}_1 \in T_1} \lambda_1(v_1,\hat{v}_1) \left( v_1 x_1( v_1 ) - p_1(v_1) - v_1 x_1( \hat{v}_1 ) + p_1(\hat{v}_1) \right) \\
& + \sum_{v_1 \in V_1}\sum_{\hat{v}_1 \in V_1} \lambda_1(v_1,\hat{v}_1) \sum_{v_2 \in V_2} f_2(v_2) \left( v_2 x_2\left( v_1, v_2 \right) - p_2\left( v_1, v_2 \right) - v_2 x_2\left( \hat{v}_1, v_2 \right) + p_2\left( \hat{v}_1, v_2 \right) \right) \\
& + \sum_{v_1\in V_1} \sum_{v_2\in V_2} \sum_{\hat{v}_2 \in V_2} \lambda_2(v_1,v_2,\hat{v}_2) \left( v_2 x_2( v_1, v_2 ) - p_2( v_1, v_2 ) -v_2 x_2( v_1, \hat{v}_2 ) + p_2( v_1, \hat{v}_2 ) \right) \\
& + \sum_{v_1\in V_1} \kappa_1(v_1) \left( v_1 x_1( v_1 ) - p_1(v_1 ) \right) + \sum_{v_1\in V_1} \sum_{v_2 \in V_2} \kappa_2(v_1,v_2) \left( v_2 x_2( v_1 , v_2 ) - p_2(v_1,v_2 ) \right)
\end{align*}

Re-grouping gives the following form:

\begin{align*}
\mathcal{L}(\lambda, \kappa, x, p)
=& \sum_{v_1\in V_1}p_1(v_1) \left( f(v_1) - \kappa_1(v_1) - \sum_{\hat{v}_1\in V_1}\lambda_1(v_1,\hat{v}_1) + \sum_{\hat{v}_1\in V_1}\lambda_1(\hat{v}_1,v_1) \right) \\
& + \sum_{v_1\in V_1}x_1(v_1) \left( v_1\kappa_1(v_1) + \sum_{\hat{v}_1\in V_1}v_1\lambda_1(v_1,\hat{v}_1) - \sum_{\hat{v}_1\in V_1}\hat{v}_1\lambda_1(\hat{v}_1,v_1) \right) \\
& + \sum_{v_1\in V_1,v_2\in V_2}p_2(v_1,v_2) \left( f(v_1)f(v_2) - \kappa_2(v_1,v_2) - \sum_{\hat{v}_2\in V_2}\lambda_2(v_1, v_2, \hat{v}_2) + \right. \\
&\qquad\qquad \left. \sum_{\hat{v}_2\in V_2}\lambda_2(v_1, \hat{v}_2, v_2) + f(v_2) \left( \sum_{\hat{v}_1\in V_1} \lambda_1(\hat{v}_1,v_1) - \sum_{\hat{v}_1\in V_1}\lambda_1(v_1,\hat{v}_1)\right)     \right)\\
&+\sum_{v_1\in V_1,v_2\in V_2}x_2(v_1,v_2) \left( v_2\kappa_2(v_1,v_2) + \sum_{\hat{v}_2\in V_2}v_2\lambda_2(v_1, v_2, \hat{v}_2) - \sum_{\hat{v}_2\in V_2}\hat{v}_2\lambda_2(v_1, \hat{v}_2, v_2) \right.\\
&\qquad\qquad \left. + v_2 f(v_2) \left( \sum_{\hat{v}_1\in V_1} \lambda_1(v_1,\hat{v}_1)- \sum_{\hat{v}_1\in V_1} \lambda_1(\hat{v}_1,v_1) \right) \right).
\end{align*}

Duality theory tells us that for any choice of $\lambda, \kappa \geq 0$, the primal objective $\opt{\mathcal{X},1,2}$ is upper bounded by $\max_{x \in \mathcal{F},p}\mathcal{L}(\lambda, \kappa, x, p)$, where $\mathcal{F}$ is the set of feasible allocations:
\begin{equation}\label{eq: opt dual upper bound}
\opt{\mathcal{X},1,2} \leq \max_{x \in \mathcal{F},p}\mathcal{L}(\lambda, \kappa, x, p)
\end{equation}

If we can find $\lambda, \kappa \geq 0$ such that $\max_{x \in \mathcal{F},p}\mathcal{L}(\lambda, \kappa, x, p)$ is bounded, we will get non-trivial upper bounds for $\opt{\mathcal{X},1,2}$. So next, we give constraints on $\lambda$ and $\kappa$ for this to be true.
Since $p_1(v_1)$ is an unconstrained variable, if its multiplier is non-zero, setting $p(v_1)$ to $\infty$ or $-\infty$ will make $\mathcal{L}(\lambda, \kappa, x, p)$ unbounded. Therefore the multipliers need to satisfy:

\begin{equation}\label{eq:stage_1_conservation}
f(v_1) - \kappa_1(v_1) - \sum_{\hat{v}_1\in V_1}\lambda_1(v_1,\hat{v}_1) + \sum_{\hat{v}_1\in V_1}\lambda_1(\hat{v}_1,v_1)  = 0 
\end{equation}

Similarly for the multiplier of $p_2(v_1,v_2)$ must be equal to $0$ for $\max_{x \in \mathcal{F},p}\mathcal{L}(\lambda, \kappa, x, p)$ to be bounded:

\begin{align}
&f(v_1)f(v_2) - \kappa_2(v_1,v_2) - \sum_{\hat{v}_2\in V_2} \lambda_2(v_1, v_2, \hat{v}_2) + \sum_{\hat{v}_2\in V_2}\lambda_2(v_1, \hat{v}_2, v_2) + f(v_2) \left( \sum_{\hat{v}_1\in V_1} \lambda_1(\hat{v}_1,v_1) - \sum_{\hat{v}_1\in V_1}\lambda_1(v_1,\hat{v}_1)\right) \notag \\
&=^{\text{Eq.~}\ref{eq:stage_1_conservation}} f(v_1)f(v_2) - \kappa_2(v_1,v_2) - \sum_{\hat{v}_2\in V_2} \lambda_2(v_1, v_2, \hat{v}_2) + \sum_{\hat{v}_2\in V_2}\lambda_2(v_1, \hat{v}_2, v_2) + f(v_2) \left( \kappa_1(v_1) - f(v_1)\right)  \notag \\
&= f(v_2)\kappa_1(v_1) - \kappa_2(v_1, v_2) - \sum_{\hat{v}_2\in V_2}\lambda_2(v_1, v_2, \hat{v}_2) + \sum_{\hat{v}_2 \in V_2}\lambda_2(v_1, \hat{v}_2, v_2) = 0 \label{eq:stage_2_conservation}
\end{align}

Similar to~\citet{cai2016} we call dual solutions that satisfy Constraints~\ref{eq:stage_1_conservation} and~\ref{eq:stage_2_conservation} \textbf{useful}. Useful dual solutions can be seen as flows in a certain tree. At the top of a tree we have a source. The nodes in the first level correspond to values in the support of the first stage; a node $v_1$ receives flow $f(v_1)$ from the source. $v_1$ can push flow $\lambda_1(v_1,v'_1)$ to some other node $v'_1$ on the same level, or push flow $\kappa_1(v_1)$ to its children. A child-node $(v_1,v_2)$, or simply $v_2$ (we explicitly say the parent when necessary), receives incoming flow $\kappa_1(v_1) \cdot f(v_2)$ from its parent. See Figure~\ref{fig:useful flow}. A similar structure is satisfied for more stages and multiple agents. We note that for correlated stages this structure fails; the incoming flow of a child-node $v_2$ depends on the flow pushed to and from its parent $v_1$.

\begin{figure}[htbp]
\centering
\begin{tikzpicture}[sibling distance=60mm, sloped,scale=0.6, every node/.style={scale=0.6}]
\node[draw,circle] (root){Source}
  child {[sibling distance=10mm] node[draw,minimum size=2em,circle,anchor=north east] (v1l) {$\downnode{v}{1}$} 
  		child{node[draw,circle,minimum size=2em,anchor=north] (a) {} edge from parent[thick,-stealth]}
  			edge from parent[thick,-stealth] 
            node[above] {$f(\downnode{v}{1})$}
  		}
  child {[sibling distance=45mm] node[draw,minimum size=2em,circle,anchor=north] (v1) {$v_1$} 
  		child{[level distance = 30mm] node[draw,circle,minimum size=2em,anchor=north] (vm) {} 
  			child {node[draw,circle,minimum size=2em,anchor=north] (v2l) {$\downnode{v}{2}$} 
  				edge from parent[thick,-stealth]
  				node[above] {$\kappa_1(v_1) \cdot f(\downnode{v}{2})$}
         		} 
  			child {node[draw,circle,minimum size=2em,anchor=north] (v2) {$v_2$}
  			    child{node[draw,circle,minimum size=2em,anchor=north] (sink) {Sink}
  			    		edge from parent[thick,-stealth]
  		 			node[above] {$\kappa_2(v_1,v_2)$}
  			    		} 
  		 		edge from parent[thick,-stealth]
  		 		node[above] {$\kappa_1(v_1) \cdot f(v_2)$}
         			}
  			child {node[draw,circle,minimum size=2em,anchor=north] (v2r) {$\upnode{v}{2}$} 
  				edge	 from parent[thick,-stealth]
  				node[above] {$\kappa_1(v_1) \cdot f(\upnode{v}{2})$}
    				}
    			edge from parent[thick,-stealth] 
            node[above] {$\kappa_1(v_1)$}
    		 }
    		edge from parent[thick,-stealth]
        node[above] {$f(v_1)$}
  }
  child {[sibling distance=20mm] node[draw,circle,anchor=north west] (v1r) {$\upnode{v}{1}$}
		  child{node[draw,circle,minimum size=2em,anchor=north] (b) {} edge from parent[thick,-stealth]} 
  			edge from parent[thick,-stealth]
  			node[above] {$f(\upnode{v}{1})$}
  };
  
\draw[thick,->] (v1) to [bend right = 10] node[draw=none,midway, below] {$\lambda_1(v_1,\upnode{v}{1})$} (v1r) ; 
\draw[thick,->] (v1r) to [bend right = 10] node[draw=none,midway, above] {$\lambda_1(\upnode{v}{1},v_1)$} (v1) ; 
\draw[thick,->] (v1l) to [bend right = 10] node[draw=none,midway, below] {$\lambda_1(\downnode{v}{1},v_1)$} (v1) ; 
\draw[thick,->] (v1) to [bend right = 10] node[draw=none,midway, above] {$\lambda_1(v_1,\downnode{v}{1})$} (v1l) ; 
\draw[thick,->] (v2) to [bend right = 5] node[draw=none,midway, below] {$\lambda_2(v_1,v_2,\upnode{v}{2})$} (v2r) ; 
\draw[thick,->] (v2r) to [bend right = 5] node[draw=none,midway, above] {$\lambda_2(v_1,\upnode{v}{2},v_2)$} (v2) ; 
\draw[thick,->] (v2l) to [bend right = 5] node[draw=none,midway, below] {$\lambda_2(v_1,\downnode{v}{2},v_2)$} (v2) ; 
\draw[thick,->] (v2) to [bend right = 5] node[draw=none,midway, above] {$\lambda_2(v_1,v_2,\downnode{v}{2})$} (v2l) ; 
\draw[thick,->] (v2l) to [bend right = 5] node[draw=none,midway, below] {$\kappa_2(v_1,\downnode{v}{2})$} (sink) ;
\draw[thick,->] (v2r) to [bend left = 5] node[draw=none,midway, below] {$\kappa_2(v_1,\upnode{v}{2})$} (sink) ;
\draw[thick,->] (a) to node[draw=none] {} (-5.85,-5) ;
\draw[thick,->] (a) to node[draw=none] {} (-6.35,-5) ;
\draw[thick,->] (a) to node[draw=none] {} (-6.85,-5) ;
\draw[thick,->] (b) to node[draw=none] {} (5.85,-5) ;
\draw[thick,->] (b) to node[draw=none] {} (6.35,-5) ;
\draw[thick,->] (b) to node[draw=none] {} (6.85,-5) ;
\end{tikzpicture}
\caption{Constraints of useful dual solutions form a flow.}\label{fig:useful flow}
\end{figure}
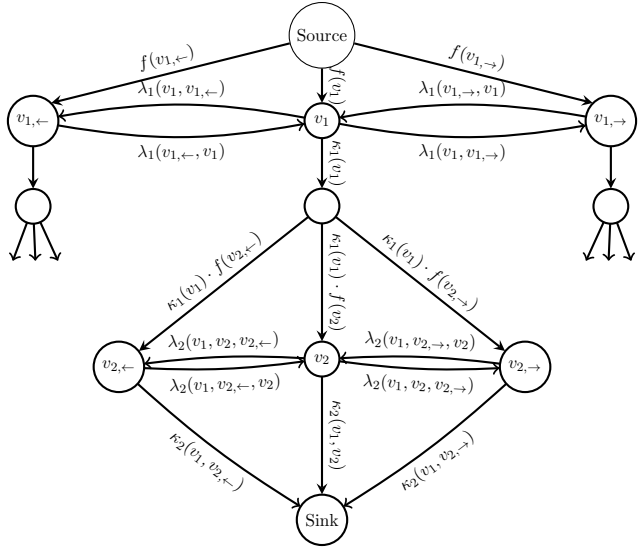

It is possible to derive familiar expressions for the multipliers of $x_1$ and $x_2$. Gathering all the terms that $x_1$ appears in $\mathcal{L}(\lambda, \kappa, x, p)$, and plugging in a useful dual solution we have:

\begin{align*}
& \sum_{v_1\in V_1} x_1(v_1)\left( v_1\kappa_1(v_1) + v_1 \sum_{\hat{v}_1\in V_1}\lambda_1(v_1,\hat{v}_1) - \sum_{\hat{v}_1\in V_1}\hat{v}_1\lambda_1(\hat{v}_1,v_1) \right) \\
&=^{\text{Eq.~}\ref{eq:stage_1_conservation}} \sum_{v_1\in V_1}x_1(v_1) \left( v_1\kappa_1(v_1) + v_1 \left( f(v_1) - \kappa_1(v_1) + \sum_{\hat{v}_1\in V_1}\lambda_1(\hat{v}_1,v_1) \right) - \sum_{\hat{v}_1\in V_1}\hat{v}_1\lambda_1(\hat{v}_1,v_1) \right) \\
&= \sum_{v_1\in V_1}x_1(v_1) f(v_1) \left( v_1 - \frac{1}{f(v_1)}\sum_{\hat{v}_1\in V_1}(\hat{v}_1-v_1) \lambda_1(\hat{v}_1,v_1) \right) \\
&= \sum_{v_1\in V_1}x_1(v_1)f(v_1)\Phi_1(v_1),
\end{align*}
where $\Phi_1(v_1) = v_1 - \frac{1}{f(v_1)}\sum\limits_{\hat{v}_1\in V_1}(\hat{v}_1-v_1) \lambda_1(\hat{v}_1,v_1)$. Therefore, every useful dual solution induces a ``virtual value'' function $\Phi_1(.)$, such that the contribution of the first stage to the $\mathcal{L}(\lambda, \kappa, x, p)$ is the expected virtual value. A similar structure is derived for the terms involving of $x_2$:

\begin{align*}
&\sum_{v_1\in V_1} \sum_{v_2\in V_2} x_2(v_1,v_2) \left( v_2 \kappa_2(v_1,v_2) + \sum_{\hat{v}_2\in V_2}v_2\lambda_2(v_1, v_2, \hat{v}_2) - \sum_{\hat{v}_2\in T_2}\hat{v}_2\lambda_2(v_1, \hat{v}_2, v_2) \right.\\
&\qquad\qquad \left. + v_2 f(v_2) \left( \sum_{\hat{v}_1\in V_1}\lambda_1(v_1,\hat{v}_1)- \sum_{\hat{v}_1\in V_1}\lambda_1(\hat{v}_1,v_1) \right) \right) \\
&=^{\text{Eq.~}\ref{eq:stage_2_conservation}} \sum_{v_1\in V_1} \sum_{v_2\in V_2}  x_2(v_1,v_2) \left( v_2 f(v_2) \kappa_1(v_1)  + \sum_{\hat{v}_2 \in V_2}\left( v_2 - \hat{v}_2 \right) \lambda_2(v_1, \hat{v}_2, v_2)  \right.\\
&\qquad\qquad \left. + v_2 f(v_2) \left( \sum_{\hat{v}_1\in V_1}\lambda_1(v_1,\hat{v}_1)- \sum_{\hat{v}_1\in V_1}\lambda_1(\hat{v}_1,v_1) \right) \right) \\
&=^{\text{Eq.~}\ref{eq:stage_1_conservation}} \sum_{v_1\in V_1} \sum_{v_2\in V_2} x_2(v_1,v_2) \left( v_2 f(v_2) \kappa_1(v_1)  + \sum_{\hat{v}_2 \in V_2}\left( v_2 - \hat{v}_2 \right) \lambda_2(v_1, \hat{v}_2, v_2) + v_2 f(v_2) \left( f(v_1) - \kappa_1(v_1) \right) \right) \\
&= \sum_{v_1\in V_1} \sum_{v_2\in V_2} x_2(v_1,v_2)\left( v_2f(v_2)f(v_1) + \sum_{\hat{v}_2\in V_2} \left( v_2 - \hat{v}_2 \right) \lambda_2(v_1, \hat{v}_2, v_2) \right) \\
&= \sum_{v_1\in V_1} \sum_{v_2\in V_2} x_2(v_1,v_2) f(v_1) f(v_2) \left( v_2 - \frac{1}{f(v_1)f(v_2)}\sum_{\hat{v}_2\in V_2}\left( \hat{v}_2 - v_2 \right) \lambda_2(v_1, \hat{v}_2, v_2) \right) \\
&= \sum_{v_1\in V_1} \sum_{v_2\in V_2} x_2(v_1,v_2)f(v_1)f(v_2)\Phi_2(v_1,v_2),
\end{align*}

where $\Phi_2(v_1,v_2) = v_2 - \frac{1}{f(v_1)f(v_2)}\sum_{\hat{v}_2\in T_2}(\hat{v}_2 - v_2)\lambda_2(v_1, \hat{v}_2, v_2)$.
Combining all the observations so far, we have that given a \textbf{useful} dual solution $\lambda, \kappa$:
\[ \mathcal{L}(\lambda, \kappa, x, p) = \sum_{v_1\in V_1}x_1(v_1)f(v_1)\Phi_1(v_1) + \sum_{v_1\in V_1} \sum_{v_2\in V_2} x_2(v_1,v_2)f(v_1)f(v_2)\Phi_2(v_1,v_2) \]

Therefore, given a useful dual solution $\lambda, \kappa$, the revenue of any dynamic mechanism $M = (x,p)$ that is ex-post IR and periodic IC, is at most the virtual welfare of $x$ with respect to the virtual value functions $\Phi_1$ and $\Phi_2$ corresponding to $\lambda$ and $\kappa$. In other words,
\[\sum_{v_1 \in V_1} f(v_1) \left( p_1(v_1) + \sum_{v_2 \in V_2} f(v_2)p_2(v_1,v_2) \right) \leq  \sum_{v_1\in V_1}f(v_1) \left( x_1(v_1)\Phi_1(v_1) + \sum_{v_2\in V_2} f(v_2) x_2(v_1,v_2) \Phi_2(v_1,v_2) \right). \]

\subsection{Canonical Flows.}

Lemma~\ref{lem:upper bound on opt one agent} is proved in two steps. First, we prove that $\opt{\mathcal{X},1,2} \leq  \expectation{  X_1 } + \myerson{ X_2 }$ in Claim~\ref{clm:e_1 plus m_2}. Then, we prove that $\opt{\mathcal{X},1,2} \leq  \expectation{  X_1 } + \myerson{ X_2 }$ in Claim~\ref{clm:m_1 plus e_2}. To prove the claims, we need the following definitions.

\begin{definition}\label{dfn: Successor and Predecessor}
For all $v^i_k \in V^i_k$ define $\upnode{v^i}{k}$ and $\downnode{v^i}{k}$ to be the values in $V^i_k$ immediately larger and immediately smaller than $v^i_k$ (respectively) :
\begin{align*}
\upnode{v^i}{k} = \inf\limits_{\hat{v}^i_k \in V^i_k: v^i_k < \hat{v}^i_k} \hat{v}^i_k && \downnode{v^i}{k} = \sup\limits_{\hat{v}^i_k \in V^i_k: v^i_k > \hat{v}^i_k} \hat{v}^i_k.
\end{align*}
\end{definition}

\begin{definition}\label{dfn: Myerson's virtual value}
Myerson's virtual value for distribution $X_k$ is \[\phi(v_k) = v_k - \frac{\left( \upnode{v}{k} - v_k \right) \cdot \Pr_{v \sim X_k} \left[ v > v_k \right] }{ f(v_k) } = v_k - \frac{\left( \upnode{v}{k} - v_k \right) \cdot \left( 1 - F(v_k) \right) }{ f(v_k) }.\]
\end{definition}

\begin{claim}\label{clm:e_1 plus m_2} 
For a single agent, two independent stages, ex-post IR and  PIC dynamic mechanisms:
\[ \opt{\mathcal{X},1,2} \leq  \expectation{  X_1 } + \myerson{ X_2 }.\]
\end{claim}

\begin{figure}[htbp]
\begin{minipage}[b][6cm][s]{.42\textwidth}
\centering
\begin{tikzpicture}[level distance = 20mm, sibling distance=60mm, sloped,scale=0.5, every node/.style={scale=0.5, solid}]
\node[draw,circle] (root){Source}
  child {[sibling distance=10mm] node[draw,minimum size=2em,circle,anchor=north east] (v1l) {$\downnode{v}{1}$} 
  		child{node[draw,circle,minimum size=2em,anchor=north] (a) {} edge from parent[thick,-stealth] node[above] {$f(\downnode{v}{1})$}}
  			edge from parent[thick,-stealth] 
            node[above] {$f(\downnode{v}{1})$}
  		}
  child {[sibling distance=45mm] node[draw,minimum size=2em,circle,anchor=north] (v1) {$v_1$} 
  		child{[level distance = 30mm] node[draw,circle,minimum size=2em,anchor=north] (vm) {} 
  			child {node[draw,circle,minimum size=2em,anchor=north] (v2l) {$\downnode{v}{2}$} 
  				edge from parent[thick,-stealth]
  				node[above] {$f(v_1) \cdot f(\downnode{v}{2})$}
         		} 
  			child {node[draw,circle,minimum size=2em,anchor=north] (v2) {$v_2$}
  			    child{node[draw,circle,minimum size=2em,anchor=north] (sink) {Sink} edge from parent[draw = none]}
  		 		edge from parent[thick,-stealth]
  		 		node[above] {$f(v_1) \cdot f(v_2)$}
         			}
  			child {node[draw,circle,minimum size=2em,anchor=north] (v2r) {$\upnode{v}{2}$} 
  				edge	 from parent[thick,-stealth]
  				node[above] {$f(v_1) \cdot f(\upnode{v}{2})$}
    				}
    			edge from parent[thick,-stealth] 
            node[above] {$f(v_1)$}
    		 }
    		edge from parent[thick,-stealth]
        node[above] {$f(v_1)$}
  }
  child {[sibling distance=20mm] node[draw,circle,anchor=north west] (v1r) {$\upnode{v}{1}$}
		  child{node[draw,circle,minimum size=2em,anchor=north] (b) {} edge from parent[thick,-stealth] node[above] {$f(\upnode{v}{1})$}} 
  			edge from parent[thick,-stealth]
  			node[above] {$f(\upnode{v}{1})$}
  };
\draw[thick,->] (v2r) to [bend right = 5] node[draw=none,midway, above] {$f(v_1) \cdot f(\upnode{v}{2})$} (v2) ; 
\draw[thick,->] (v2) to [bend right = 5] node[draw=none,midway, above] {$f(v_1) \cdot f(\upnode{v}{2})$} (v2l) ;
\draw[thick,->] (v2) to [bend right = 5] node[draw=none,midway, below] {$ + f(v_1) \cdot f(v_2)$} (v2l) ; 
\draw[thick,->] (v2l) to [bend right = 5] node[draw=none,midway, below] {$f(v_1)$} (sink) ;
\draw[thick,->] (a) to node[draw=none] {} (-5.85,-6) ;
\draw[thick,->] (a) to node[draw=none] {} (-6.35,-6) ;
\draw[thick,->] (a) to node[draw=none] {} (-6.85,-6) ;
\draw[thick,->] (b) to node[draw=none] {} (5.85,-6) ;
\draw[thick,->] (b) to node[draw=none] {} (6.35,-6) ;
\draw[thick,->] (b) to node[draw=none] {} (6.85,-6) ;
\end{tikzpicture}
\caption{An example with support $3$ of the flow with Lagrangian $\expectation{X_1} + \myerson{X_2}$.}\label{fig: expectation plus myerson flow}
\end{minipage}\hfill
\begin{minipage}[b][6cm][s]{.42\textwidth}
\centering
\begin{tikzpicture}[level distance = 20mm, sibling distance=60mm, sloped,scale=0.5, every node/.style={scale=0.5, solid}]
\node[draw,circle] (root){Source}
  child {[sibling distance=10mm] node[draw,minimum size=2em,circle,anchor=north east] (v1l) {$\downnode{v}{1}$} 
  		child{node[draw,circle,minimum size=2em,anchor=north] (a) {} edge from parent[-stealth] node[above] {$1$}}
  			edge from parent[thick,-stealth] 
            node[above] {$f(\downnode{v}{1})$}
  		}
  child {[sibling distance=20mm] node[draw,minimum size=2em,circle,anchor=north] (v1) {$v_1$} 
  		child{[level distance = 30mm] node[draw,circle,minimum size=2em,anchor=north] (vm) {} 
  			child {node[draw,circle,minimum size=2em,anchor=north] (v2l) {$\downnode{v}{2}$} } 
  			child {node[draw,circle,minimum size=2em,anchor=north] (v2) {$v_2$}
  			    child{node[draw,circle,minimum size=2em,anchor=north] (sink) {Sink} edge from parent[draw = none] } 
  		 		edge from parent[dashed]}
  			child {node[draw,circle,minimum size=2em,anchor=north] (v2r) {$\upnode{v}{2}$} 
  				edge	 from parent[dashed]}
    			edge from parent[dashed] 
            node[right, rotate = 90] {$\kappa_1(v_1) = 0$}
    		 }
    		edge from parent[thick,-stealth]
        node[above] {$f(v_1)$}
  }
  child {[sibling distance=20mm] node[draw,circle,anchor=north west] (v1r) {$\upnode{v}{1}$}
		  child{node[draw,circle,minimum size=2em,anchor=north] (b) {} edge from parent[dashed]} 
  			edge from parent[thick,-stealth]
  			node[above] {$f(\upnode{v}{1})$}
  };
\node[draw,circle] (random1) at (-6.85,-6) {};
\node[draw,circle] (random2) at (-6.35,-6) {};
\node[draw,circle] (random3) at (-5.85,-6) {};
\draw[thick,->] (v1r) to [bend right = 10] node[draw=none,midway, above] {$f(\upnode{v}{1})$} (v1) ; 
\draw[thick,->] (v1) to [bend right = 10] node[draw=none,midway, above] {$f(\upnode{v}{1}) + f(v_1)$} (v1l) ; 
\draw[thick,->] (a) to node[draw=none] {} (random3) ;
\draw[thick,->] (a) to node[draw=none] {} (random2) ;
\draw[thick,->] (a) to node[draw=none] {} (random1) ;
\draw[dashed,->] (b) to node[draw=none] {} (5.85,-6) ;
\draw[dashed,->] (b) to node[draw=none] {} (6.35,-6) ;
\draw[dashed,->] (b) to node[draw=none] {} (6.85,-6) ;
\draw[thick,->] (random1) to [bend right = 30] node[draw=none] {} (sink) ;
\draw[thick,->] (random2) to [bend right = 30] node[draw=none] {} (sink) ;
\draw[thick,->] (random3) to [bend right = 30] node[draw=none] {} (sink) ;
\end{tikzpicture}
\caption{An example with support $3$ of the flow with Lagrangian $\myerson{X_1} + \expectation{X_2}$.}\label{fig: myerson plus expectation flow}
\end{minipage}
\end{figure}
\vspace{2cm}

\begin{proof}
Consider the following dual solution:
\begin{align*}
\kappa_1(v_1) &= f(v_1) && \lambda_1(v_1, \hat{v}_1) = 0
\end{align*}
\begin{align*}
\kappa_2(v_1,v_2) = 
\begin{cases}
f(v_1) \quad& \text{ if } v_2 = \underline{v}_2 \\
0 & \text{ o.w.}
\end{cases}
&& 
\lambda_2 (v_1, v_2, \hat{v}_2) = 
\begin{cases}
f(v_1)(1 - F(\hat{v}_2)) & \text{ if } \hat{v}_2 = \downnode{v}{2} \\
0 & \text{ o.w.} \\
\end{cases}
\end{align*}

It's easy to verify that constraints \ref{eq:stage_1_conservation} and \ref{eq:stage_2_conservation} are satisfied; the solution is useful. Figure~\ref{fig: expectation plus myerson flow} illustrates the dual solution's corresponding flow for distributions with support 3. In the first level, a node $v_1$ receives the flow $f(v_1)$ from the source and pushes the flow $\kappa_1(v_1) = f(v_1)$ to its children. There is no flow between nodes in the first level as $\lambda_1(v_1, \hat{v}_1) = 0$. In the second level, a child-node $(v_1,v_2)$ receives the flow $\kappa(v_1)f(v_2) = f(v_1)f(v_2)$ from its parent, and the flow $f(v_1)(1 - F(v_2))$ from $(v_1,\upnode{v}{2})$. It pushes all its flow $f(v_1)f(v_2) + f(v_1)(1 - F(v_2)) = f(v_1)(1 - F(\downnode{v}{2}))$ to $(v_1, \downnode{v}{2})$ or to the sink when $(v_1, \downnode{v}{2})$ does not exist.
These flows induce virtual values $\Phi_1(v_1) = v_1$ for the first stage nodes. For the second stage nodes, $\Phi_2(v_1,v_2)$ becomes equal to $\phi(v_2)$, Myerson's virtual value for $X_2$. For simplicity we assume that $X_2$ is regular, i.e. the virtual value $\Phi_2$ induced by our flow is monotone non-decreasing; if this is not the case we can ``iron'' our flow by adding loops (See the ironing procedure in~\cite{cai2016}).
By Equation~\ref{eq: opt dual upper bound}, 

\begin{align*}
\opt{\mathcal{X},1,2} &\leq \max_{x,p}\mathcal{L}(\lambda, \kappa, x, p) \\
&= \max_{x,p} \sum_{v_1\in V_1}x_1(v_1)f(v_1)\Phi_1(v_1) + \sum_{v_1\in V_1,v_2\in V_2}x_2(v_1,v_2)f(v_2)f(v_1)\Phi_2(v_1,v_2) \\
&= \max_{x,p} \sum_{v_1\in V_1}x_1(v_1)f(v_1)v_1 \\
&\qquad + \sum_{v_1\in V_1,v_2\in V_2}x_2(v_1,v_2)f(v_2)f(v_1)\left( v_2 - \frac{1}{f(v_1)f(v_2)}\sum_{\hat{v}_2\in V_2}(\hat{v}_2 - v_2)\lambda_2(v_1, \hat{v}_2, v_2) \right) \\
&= \expectation{  X_{1} } + \max_{x,p} \sum_{v_1\in V_1,v_2\in V_2}x_2(v_1,v_2)f(v_2)f(v_1)\left( v_2 - \frac{1}{f(v_1)f(v_2)}(\upnode{v}{2} - v_2)f(v_1)\left( 1 - F(v_2) \right) \right) \\
&= \expectation{  X_{1} } + \sum_{v_1\in V_1} f(v_1) \max_{x,p} \sum_{v_2\in V_2}x_2(v_1,v_2)f(v_2)(v_2 - \frac{1 - F(v_2)}{f(v_2)}(\upnode{v}{2} - v_2)) \\
&= \expectation{  X_{1} } + \sum_{v_1\in V_1} f(v_1) \max_{x,p} \sum_{v_2\in V_2}x_2(v_1,v_2)f(v_2) \phi(v_2) \\
&= \expectation{  X_{1} } + \sum_{v_1\in V_1} f(v_1) \myerson{ X_2 } \\
&= \expectation{  X_{1} } + \myerson{ X_2 } \qedhere\\
\end{align*}
\end{proof}

\begin{claim}\label{clm:m_1 plus e_2} 
For single agent, two independent stages, ex-post IR and  PIC dynamic mechanisms:
\[ \opt{\mathcal{X},1,2} \leq  \myerson{ X_1 } + \expectation{  X_{2} } \]
\end{claim}

\begin{proof}
Similar to the proof of Claim~\ref{clm:e_1 plus m_2}. Start with an assignment of $\lambda, \kappa$:
\begin{align*}
\kappa_1(v_1) = 
\begin{cases}
1 \quad& \text{ if } v_1 = \underline{v}_1  \\
0  &\text{ o.w.} \\
\end{cases}
&&
\lambda_1 (v_1, \hat{v}_1) = 
\begin{cases}
1 - F(\hat{v}_1)\quad & \text{ if } \hat{v}_1 = \downnode{v}{1}  \\
0  &\text{ o.w.} \\
\end{cases}
\end{align*}

\begin{align*}
\kappa_2(v_1,v_2) = 
\begin{cases}
f(v_1)f(v_2) \quad& \text{ if } v_1 = \underline{v}_1  \\
0 &\text{ o.w.} \\
\end{cases}
&&&&
\lambda_2 (v_1, v_2, \hat{v}_2) = 0
\end{align*}

It's easy to verify that Constraints~\ref{eq:stage_1_conservation} and \ref{eq:stage_2_conservation} are satisfied. Figure~\ref{fig: myerson plus expectation flow} illustrates the dual solution's corresponding flow for distributions with support 3. In the first level, a node $v_1$ receives the flow $f(v_1)$ from the source and the flow $1- F(v_1)$ from $\upnode{v}{1}$. It pushes the flow $\lambda_1(v_1,\downnode{v}{1}) = f(v_1) + 1 - F(v_1) = 1- F(\downnode{v}{1})$ to $\downnode{v}{1}$. If $v_1 = \underline{v}_1$, $\downnode{v}{1}$ does not exist for $v_1$. Then $\underline{v}_1$ pushes $\kappa_1(v_1) = 1$ to its children. In the second level, a child-node $(v_1,v_2)$ receives no flow from the parent unless $v_1 = \underline{v}_1$. The child-node $(\underline{v}_1, v_2)$ pushes its incoming flow $f(v_2)$ to the sink.
This time, $\Phi_1(v_1)$ is Myerson's virtual value and $\Phi_2(v_1,v_2) = v_2$, for every node $v_2$, expect the children of $\underline{v}_1$. Note that, given $\kappa_1(v_1) = 0$, the every child-node under $v_1$ has no incoming flow; therefore, it is unavoidable for their virtual value to be equal to their value.

\begin{align*}
\opt{\mathcal{X},1,2} &\leq \max_{x,p}\mathcal{L}(\lambda, \kappa, x, p) \\
&= \max_{x,p}\sum_{v_1\in V_1}x_1(v_1)f(v_1)\Phi_1(v_1) + \sum_{v_1\in V_1,v_2\in V_2}x_2(v_1,v_2)f(v_2)f(v_1)\Phi_2(v_1,v_2) \\
&=  \max_{x,p} \sum_{v_1\in V_1}x_1(v_1)f(v_1) \left( v_1 - \frac{1-F(v_1)}{f(v_1)}(\upnode{v}{1} - v_1) \right) + \sum_{v_1\in V_1}\sum_{v_2\in V_2}x_2(v_1,v_2)f(v_1)f(v_2)v_2 \\
&= \max_{x,p} \sum_{v_1\in V_1}x_1(v_1)f(v_1)\phi(v_1) + \expectation{ X_{2} } \\
&= \myerson{ X_1 } + \expectation{ X_{2} } \qedhere\\
\end{align*}
\end{proof}

Claim~\ref{clm:e_1 plus m_2} and Claim~\ref{clm:m_1 plus e_2} together imply Lemma~\ref{lem:upper bound on opt one agent}.

\section{Lower Bounding the Revenue of VCG}\label{sec:second price bounds}

In this Section we prove Theorem~\ref{thm: bounds on mhr stuff}. The proof is broken into three Lemmas.
Recall that the hazard rate of a distribution $F$ is $h(x) = \frac{f(x)}{1-F(x)}$. $F$ has monotone hazard rate (MHR) if $h(x)$ is a non-decreasing function. We restrict ourselves to continuous distributions. For Lemma~\ref{lem:general order statistics} we also need the distribution to be supported on $[0,\bar{V}]$. 

\begin{lemma}\label{lem:general order statistics}
Let $X_{r:n}$ be the $r$-th order statistic of $n$ i.i.d. samples from a continuous (possibly unbounded) distribution with monotone hazard rate. Then $4n$ samples are necessary and sufficient for $\expectation{X_{2:4n}}$ to be at least as large as $\expectation{X_{1:n}}$.
\end{lemma}

\begin{lemma}\label{lem: oneovere}
Let $X_{r:n}$ be the $r$-th (highest) order statistic of $n$ i.i.d. samples from a continuous (possibly unbounded) distribution with monotone hazard rate. Then $\expectation{X_{2:n+1}} \geq \frac{1}{e} \expectation{X_{1:n}}$.
\end{lemma}

\begin{lemma}\label{lem: oneoverthree}
Let $X_{r:n}$ be the $r$-th (highest) order statistic of $n$ i.i.d. samples from a continuous (possibly unbounded) distribution with monotone hazard rate. Then $\expectation{X_{2:n}} \geq \frac{1}{3} \expectation{X_{1:n}}$.
\end{lemma}

A useful fact about order statistics of MHR distributions that we use throughout this Section is that order statistics of MHR distributions have themselves an MHR distribution:

\begin{lemma}[\citet{barlow1996mathematical}]\label{lem:barlow order statistics of MHR}
Assume $X$ is a random variable with distribution $F$ and density $f$ which is MHR. If $X_1 , X_2, \dots, X_n$, are $n$ independent observations on $X$, the order statistics formed from the $X_i$'s are also MHR.
\end{lemma}

We break the proof of Lemma~\ref{lem:general order statistics} into two parts. We first prove the result for $n=1$ in Subsection~\ref{subsec: single sample strict mhr bound}. We complete the proof of general $n$ by combining the $n=1$ case, Lemma~\ref{lem:barlow order statistics of MHR} and a coupling argument in Subsection~\ref{subsec: general mhr bound}.
We prove Lemmas~\ref{lem: oneovere} and~\ref{lem: oneoverthree} (necessary for our $\frac{1}{e}$ and $\frac{1}{3}$-approximate Competition Complexity bounds) in Subsections~\ref{subsec: one sample more} and~\ref{subsec: no samples} respectively. 

\subsection{Bound when $n=1$}\label{subsec: single sample strict mhr bound}


\begin{lemma}\label{lemma: single sample}
Let $X_{r:n}$ the $r$-th order statistic of $n$ i.i.d. samples from a continuous distribution with monotone hazard rate. Then $4$ samples are necessary and sufficient for $\expectation{X_{2:4}}$ to be at least as large as $\expectation{X} = \expectation{X_{1:1}}$.
\end{lemma}

Let $H(x) = \int_0^x h(z) dz$. If $F$ is MHR, then $H(x)$ is a convex function as it is the integral of a non-decreasing function. 
The proofs of the next Claims can be found in Appendix~\ref{appendix:missing from mhr}.

\begin{claim}\label{claim: basic MHR }
$F(x) = 1 - e^{-H(x)}$ and $\expectation{X} = \int_0^{\bar{V}} e^{-H(x)} dx$.
\end{claim}


\begin{claim}\label{claim:expected second max formula}
$\expectation{X_{2:4}} = \int_0^{\bar{V}} 3 e^{-4H(x)} - 8 e^{-3H(x)} + 6 e^{-2H(x)} dx$.
\end{claim}

For the upper bound on the number of samples, it suffices to show that $\int_0^{\bar{V}} 3 e^{-4H(x)} - 8 e^{-3H(x)} + 6 e^{-2H(x)} - e^{-H(x)} dx \geq 0$ for all non-negative, convex and continuous functions $H(x)$. We first prove this statement for all non-negative, piecewise linear and convex functions $\hat{H}(x)$ in Lemma~\ref{lemma: piecewise linear H}. We then show how to approximate any convex function by a piecewise linear convex function in Lemma~\ref{lemma: piecewise linear approximation}. We combine Lemmas~\ref{lemma: piecewise linear H} and~\ref{lemma: piecewise linear approximation} to prove the upper bound in Lemma~\ref{lemma: single sample}. The lower bound comes from considering a (truncated) exponential distribution.

The proof of the following lemma can be found in Appendix~\ref{appendix:missing from mhr}.

\begin{lemma}\label{lemma: piecewise linear H}
Let $\hat{H}(x)$ be a non-negative, piecewise linear and convex function in $[0,\bar{V}]$, with $\hat{H}(0) = 0$. Then $\int_0^{\bar{V}} 3 e^{-4\hat{H}(x)} - 8 e^{-3\hat{H}(x)} + 6 e^{-2\hat{H}(x)} - e^{-\hat{H}(x)} dx > 0.$
\end{lemma}

Our next step is to show that any convex function $H(x)$ can be approximated by a piecewise linear and convex function $\hat{H}(x)$. An elementary theorem from Real Analysis tells us that for every continuous function $f(x)$ in a closed interval $[a,b]$ and every $\epsilon > 0$, there exists a piecewise linear function $g_{\epsilon}(x)$ such that $\forall x \in [a,b], \left| f(x) - g_{\epsilon}(x) \right| < \epsilon$. We show that the same is true for every continuous and convex function $f(x)$, $\epsilon > 0$, where this time the approximation is by a piecewise linear and convex function $g(x)$\footnote{We don't believe this Lemma to be new.}. The proof can be found in Appendix~\ref{appendix:missing from mhr}.

\begin{lemma}\label{lemma: piecewise linear approximation}
For every function $f(x)$ that is continuous and convex in a closed interval $[a,b]$, $\epsilon > 0$, there exists a convex piecewise linear function $g_{\epsilon}(x)$ such that for all $x \in [ a, b ]$: $ \left| f(x) - g_{\epsilon}(x) \right| < \epsilon $.
\end{lemma}

We are now ready to prove Lemma~\ref{lemma: single sample}:

\begin{proof}[Proof of Lemma~\ref{lemma: single sample}]
First, we show that if there exists an unbounded distribution $X$ that is MHR, such that $\expectation{X_{2:4}} - \expectation{X_{1:1}} = -\delta'$\footnote{We know that this difference exists because order statistics of MHR distributions are finite.}, for some $\delta' > 0$, then there exists a bounded distribution $Y$ that is MHR, with property $\expectation{Y_{2:4}} - \expectation{Y_{1:1}} = - \frac{\delta'}{2} = -\delta$. We get $Y$ by truncating $X$, therefore $\expectation{Y_{2:4}} \leq \expectation{X_{2:4}}$:
\begin{align*}
- \delta' = \expectation{X_{2:4}} - \expectation{X_{1:1}} \geq \expectation{Y_{2:4}} - \expectation{X_{1:1}}.
\end{align*}

$\expectation{X} = \int_{x=0}^{\infty} x f_X(x) dx = \int_{x=0}^{V} x f_X(x) dx + \int_{x=V}^\infty x f_X(x) dx = \expectation{Y} + \int_{x=V}^\infty x f_X(x) dx$. From a bound of~\cite{cai2011extreme}, we know that $\int_{x=V}^\infty x f_X(x) dx \leq \frac{6 \alpha_m}{m}$, where $\alpha_m = \inf\{ x | F_X(x) \geq 1 - \frac{1}{m} \}$. Furthermore, again from a Lemma of~\cite{cai2011extreme}, $k \alpha_m \geq \alpha_{m^k}$, which implies that as $m \rightarrow \infty$, $\int_{x=V}^\infty x f_X(x) dx \rightarrow 0$. Therefore, we can pick an $m^*$ large enough, such that $\int_{x=V}^\infty x f_X(x) dx$ is at most $\frac{\delta'}{2}$. Combining, we get that there exists a truncation point for $X$ such that for the truncated distribution $Y$\footnote{Truncated MHR distributions are MHR.}  $\expectation{Y_{2:4}} - \expectation{Y} \geq -\delta$.

Then there exists a non-negative convex and continuous function $H(y)$ in $[0,m^*]$ such that $$ 
I\left[ H \right] = \int_0^{m^*} 3 e^{-4H(y)} - 8 e^{-3H(y)} + 6 e^{-2H(y)} - e^{-H(y)} dy = - \delta,$$ for some $\delta > 0$. Let $\hat{H}_{\epsilon}(y)$ be a piecewise linear and convex function that $\epsilon$-approximates $H(y)$, as in Lemma~\ref{lemma: piecewise linear approximation}. Since $I\left[ H \right]$ is a bounded integral, we can choose $\epsilon$ small enough such that
\[  \left| I\left[ H \right] - I\left[  \hat{H}_{\epsilon} \right] \right| < \delta. \]
This would imply that $\int_0^{m^*} 3 e^{-4\hat{H}_{\epsilon}(y)} - 8 e^{-3\hat{H}_{\epsilon}(y)} + 6 e^{-2\hat{H}_{\epsilon}(y)} - e^{-\hat{H}_{\epsilon}(y)} dy < 0$, a contradiction to Lemma~\ref{lemma: piecewise linear H}. Therefore, $4$ samples are sufficient for $\expectation{X_{2:4}} \geq \expectation{X_{1:1}}$. 

To see that $4$ samples are also necessary, consider the second order statistic from $3$ samples with cdf $F_{2:3} = 3 F(x)^{2} - 2 F(x)^3$ and expectation $\expectation{X_{2:3}} = \int_0^{\bar{V}} 3e^{-2H(x)} - 2e^{-3H(x)} dx$. An exponential distribution with parameter $\lambda = 1$ gives $\expectation{X_{2:3}} = \frac{5}{6} < 1 = \lambda = \expectation{X} = \expectation{X_{1:1}}$, but is not bounded. We can truncate at some large $\bar{V}$ in a way that neither $\expectation{X_{2:4}}$ nor $\expectation{X}$ change by more than a negligible amount (truncated exponential distributions still have non-decreasing hazard rate). Therefore, $4$ samples are also necessary.
\end{proof}

\subsection{Proof of Lemma~\ref{lem:general order statistics}}\label{subsec: general mhr bound}
\begin{proof}

We have already shown (for $n=1$) that $4n$ samples are necessary in Lemma~\ref{lemma: single sample}. Therefore, it remains to show that $4n$ samples are sufficient. Let $Y = X_{1:n}$ be the maximum of $n$ i.i.d. samples from $F$. Let the $F_y$ be the cdf of $Y$. If $F$ is MHR, then so is $F_y$ (Lemma~\ref{lem:barlow order statistics of MHR}).

Since $Y$ is MHR, by Lemma~\ref{lemma: single sample} we have that $\expectation{Y_{2:4}} \geq \expectation{Y} = \expectation{X_{1:n}}$, where $Y_{2:4}$ is the second order statistic of $4$ samples drawn from $F_y$. Therefore, it suffices to show that $\expectation{X_{2:4n}} \geq \expectation{Y_{2:4}}$: Draw $4n$ samples $X_1,\dots,X_{4n}$ from $F$. Let $Z_1 = \max\limits_{i=1:n} X_i , Z_2 = \max\limits_{i=n+1:2n} X_i , Z_3 = \max\limits_{i=2n+1:3n} X_i$ and $Z_4 = \max\limits_{i=3n+1:4n} X_i$. $Y_{2:4}$ is the second largest of the $Z_i$'s. On the other hand, $X_{2:4n}$ is the second largest of the $X_i$'s, and therefore at least as large as $Y_{2:4}$, for every single outcome.\end{proof}


\subsection{Bounding $\expectation{X_{2:n+1}}$: Proof of Lemma~\ref{lem: oneovere}}\label{subsec: one sample more}


We are going to use the following technical lemma:
\begin{lemma}\label{lem: myerson for n and one}
Let $X$ be a random variable from an MHR distribution $D$. Let $X_{1:n}$ be the largest order statistic of $n$ samples from $D$, and let $\mathbf{X} = \Pi_{i=1}^n X$ denote the product distribution of $n$ agents. Then $\opt{\mathbf{X}} \geq \opt{X_{1:n}}$, i.e. the optimal revenue of $n$ i.i.d. agent from $X$ is larger than the optimal revenue of one agent from $X_{1:n}$.
\end{lemma}

\begin{proof}
The optimal auction $M$ on distribution $\mathbf{X}$ is a second price auction with some reserve $p$. The optimal auction $M'$ on $X_{1:n}$ is a posted price auction, with some posted price $p'$. 
We can calculate the revenue of $M'$ on $X_{1:n}$ as follows: first draw  $v_1,v_2,\dots,v_n$ be $n$ i.i.d. samples from $X$, i.e. a sample $\mathbf{v}$ from $\mathbf{X}$. If the maximum of the $v_i$'s is larger than $p'$, then the revenue of $M'$ on this outcome is $p'$: $\opt{X_{1:n}} = \sum_{ \mathbf{v} \sim \mathbf{D} } \Pr[\mathbf{v}] \cdot p' \cdot \mathds{1}\left[ \max_i v_i \geq p' \right]$.

Let $\hat{M}$ be a second price auction with reserve $p'$. If the maximum of the $v_i$'s is larger than $p'$, then the revenue of $\hat{M}$ on this outcome is the larger of $p'$ and the second largest $v_i$. Therefore $\rev{ \hat{M},\mathbf{X} } \geq \opt{X_{1:n}}$. But, $M$ is the second price auction with the \textit{optimal} posted price, and  therefore $\opt{\mathbf{X}} = \rev{ M, \mathbf{X} }\geq \rev{ \hat{M}, \mathbf{X} }$. The Lemma follows.
\end{proof}

\begin{proof}[Proof of Lemma~\ref{lem: oneovere}]
First, since $X$ is a random variable from an MHR distribution $D$, we can lower bound $\expectation{X_{2:n+1}}$ (the revenue of a second price auction) using the original Theorem of Bulow and Klemperer(\cite{bulow1996auctions}):
\begin{equation}\label{eq:bk}
\expectation{X_{2:n+1}} = \rev{ \text{Vickrey with $n+1$ i.i.d. agents from D} } \geq \opt{ \text{$n$ i.i.d. agents from D} }.
\end{equation}

Second, by Lemma~\ref{lem: myerson for n and one}:
\begin{equation}\label{eq:n better than one}
\opt{ \text{$n$ i.i.d. agents from D} } \geq \opt{ \text{one agent from $X_{1:n}$} }.
\end{equation}

Third, order statistics of MHR distributions have MHR distributions themselves (Lemma~\ref{lem:barlow order statistics of MHR}) , i.e. $X_{1:n}$ has monotone hazard rate. Fourth, from a known result from auction theory (e.g.~\citet{hartline2013mechanism}, Lemma 5.14) we have that the optimal expected revenue from an MHR distribution is an $e$ approximation to the optimal expected surplus. Applying to $X_{1:n}$ gives:
\begin{equation}\label{eq:e approximation}
\opt{X_{1:n}} \geq \frac{1}{e} \expectation{X_{1:n}}.
\end{equation}

Combining Equations~\ref{eq:bk},~\ref{eq:n better than one} and~\ref{eq:e approximation} gives the Lemma.
\end{proof}

\subsection{Bounding $\expectation{X_{2:n}}$: Proof of Lemma~\ref{lem: oneoverthree}}\label{subsec: no samples}

\begin{proof}

We need the following two Lemmas. The first Lemma is proved in~\cite{blog}. 
The proof uses that $F_{2:2}(x) = 1-(1-F(x))^2 = 1-e^{-2H(x)}$ and the fact that $H(x)$ is convex. 
The second Lemma was proved by~\citet{barlow1966inequalities}\footnote{Also see~\citet{szech2011optimal}}. We include the proofs in Appendix~\ref{appendix:missing from mhr}.

\begin{lemma}\label{lem:expected min of two samples}
$\expectation{X_{1:2}} - \expectation{X_{2:2}} \leq \frac{2}{3} \expectation{X_{1:2}}$.
\end{lemma}

\begin{lemma}\label{lem: expected spacing is decreasing}
$\expectation{X_{1:n}} - \expectation{X_{2:n}} = \expectation{ \frac{1}{h(X_{1:n})} }$, where $h(x) = \frac{f(x)}{1-F(x)}$, and thus is a non-increasing function of $n$ for MHR distributions.
\end{lemma}

Given the two Lemmas above, we get the desired bound on $\expectation{X_{2:n}}$ as follows:
\[ \expectation{X_{1:n}} - \expectation{X_{2:n}} \leq \expectation{X_{1:2}} - \expectation{X_{2:2}} \leq \frac{2}{3} \expectation{X_{1:2}} \leq \frac{2}{3} \expectation{X_{1:n}} \qedhere \] 
\end{proof}

\section{Revenue Non-Monotonicity.}\label{sec:non monotonicity}

\subsection{Correlation Increases Revenue}

Let the first stage and second stage distributions be:

\vspace{-4mm}

\noindent\begin{minipage}[t]{.5\textwidth}
\begin{gather*}
X_1 = 
\begin{cases} 
2^{n+1} &\mbox{w.p. }  \frac{1}{2} = 2^{-1} \\
2^{n+2} & \mbox{w.p. } \frac{1}{4} = 2^{-2} \\
2^{n+3} & \mbox{w.p. } 2^{-3} \\
\dots \\
2^{n+(n-1)} = 2^{2n-1} & \mbox{w.p. } 2^{-(n-1)} \\
2^{2n} & \mbox{w.p. } 2^{-n} \\
2^{2n+1} & \mbox{w.p. } 2^{-n} \\
\end{cases} 
\end{gather*}
\end{minipage}%
\begin{minipage}[t]{.5\textwidth}
\begin{gather*}
X_2 = 
\begin{cases} 
2^{n+1} &\mbox{w.p. }  \frac{1}{2} = 2^{-1} \\
2^{n} & \mbox{w.p. } \frac{1}{4} = 2^{-2} \\
2^{n-1} & \mbox{w.p. } 2^{-3} \\
\dots \\
2^{n+1 - (n-1)} = 8 & \mbox{w.p. } 2^{-(n-1)} \\
2^{n+1 - (n-1)} = 4 & \mbox{w.p. } 2^{-n} \\
2^{n+1 - n} = 2 & \mbox{w.p. } 2^{-n} \\
\end{cases} 
\end{gather*}
\end{minipage}%

\noindent$X_1$ is an equal revenue distribution and therefore $\myerson{ X_1 } = 2^{n+1}$ (from setting any price). Also:

\vspace{-4mm}

\begin{talign*}
\expectation{X_2} &= 2^{n+1} \cdot \frac{1}{2} + 2^n \cdot \frac{1}{4} + 2^{n-1} \cdot 2^{-3} + \dots + 2^3 \cdot 2^{-(n-1)} + 2^2 \cdot 2^{-n} + 2 \cdot 2^{-n} \\
&= 2^n + 2^{n-2} + 2^{n-4} + 2^{n-6} + \dots + 2^{-n+4} + 2^{-n+2} + 2^{-n+1} \\
&= 2^{-n+1} + \sum_{i=1}^n 2^{n+2-2i} \\
&= 2^{-n+1} + \frac{1}{3} \cdot 2^{-n+2} \cdot (2^{2n} - 1) \\
&= \frac{2}{3} \left( 2^{-n} + 2^{n+1} \right)
\end{talign*}

Therefore, the revenue of the optimal auction for independent $X_1$ and $X_2$ is at most:
\[ \textstyle
\myerson{ X_1 } + \expectation{X_2} = 2^{n+1} + \frac{2}{3} \left( 2^{-n} + 2^{n+1} \right) = 2^{n+1} + \frac{4}{3} 2^{n} + \frac{2}{3} 2^{-n}.
\]

Next we consider correlated $X_1$ and $X_2$, while keeping the marginals for each stage the same, and show an auction with revenue strictly larger than $2^{n+1} + \frac{4}{3} 2^{n} + \frac{2}{3} 2^{-n}$. The correlation is as follows:
If the value on the first stage is $2^{n+1}$ then the value on the second stage is $2^{n+1}$ with probability $1$. Otherwise, if the first stage value is $2^{n+i}$ for $i=2,\dots,n+1$, the value on the second stage is $2^{n+2-i}$ with probability $1$ ($2^{n+2}$ maps to $2^n$, $2^{n+3}$ maps to $2^{n-1}$ and so on).

Consider the auction that in the first stage offers two options: (1) Pay $2^n$ and get the first day item with probability $\frac{1}{2}$. The second stage item will cost $2^{n+1}$; (2) Pay $2^{n+2}$ and get the first day item with probability $1$. The second stage item will be free.

The buyer with first stage value $2^{n+1}$ will buy the first option. Everyone else will buy the second option.
The total revenue is $\frac{1}{2} \cdot \left( 2^n + 2^{n+1} \right) + \frac{1}{2} \cdot 2^{n+2} = 2^{n-1} + 2^n + 2^{n+1}$,
which is larger than $2^{n+1} + \frac{4}{3} 2^{n} + \frac{2}{3} 2^{-n}$ for all $n \geq 2$.

\subsection{Correlation Decreases Revenue}

Let the first stage and second stage distributions be:

\vspace{-5mm}

\noindent\begin{minipage}[t]{.5\textwidth}
\begin{gather*}
X_1 = 
\begin{cases} 
1  &\mbox{w.p. }  \frac{1}{2} = 2^{-1} \\
2 & \mbox{w.p. } \frac{1}{4} = 2^{-2} \\
3 & \mbox{w.p. } 2^{-3} \\
\dots \\
n-1 & \mbox{w.p. } 2^{-(n-1)} \\
n & \mbox{w.p. } 2^{-n} \\
n+1 & \mbox{w.p. } 2^{-n} \\
\end{cases} 
\end{gather*}
\end{minipage}%
\begin{minipage}[t]{.5\textwidth}
\begin{gather*}
X_2 = 
\begin{cases} 
2 &\mbox{w.p. }  \frac{1}{2} = 2^{-1}  \\
2^{2} & \mbox{w.p. } \frac{1}{4} = 2^{-2} \\
2^{3} & \mbox{w.p. } 2^{-3} \\
\dots \\
2^{n-1} & \mbox{w.p. } 2^{-(n-1)} \\
2^{n} & \mbox{w.p. } 2^{-n} \\
2^{n+1} & \mbox{w.p. } 2^{-n} \\
\end{cases} 
\end{gather*}
\end{minipage}%

If $X_1$ and $X_2$ are independent, the optimal auction extracts revenue equal to $\expectation{X_1} + \myerson{X_2}$. $X_2$ is an equal revenue distribution, so, as we've already seen, $\myerson{X_2} = 2$ and is achieved by setting any price $p$. Another interesting property of $X_2$ is that even though setting different prices $p$ has the same revenue for the seller, the result is different expected utilities for the buyer. For example, setting $p=2$ gives expected utility for the buyer $\expectation{X_2} - 2 = n$, setting price $p = 4$ gives expected utility $\sum_{v \geq 4} Pr[ X_2 = v ] \left( v - 4 \right) = n-1$, setting price $p=2^3$ gives expected utility $n-2$, and so on. The following auction is truthful (in a perfect Bayesian equilibrium) and ex-post IR: On the first stage the price is equal to the reported value $\hat{v}$. For a reported value $\hat{v}$ the second stage item will have a posted price $p$ that results in expected utility $\hat{v} - 1$.
Observe that reporting value $\hat{v}$ gives utility $v - \hat{v}$ for the first stage, and utility $\hat{v} - 1$ for the second stage, i.e. overall utility $v-1$. The revenue of this auction is $\expectation{X_1} + \myerson{X_2} > 3$.

Now, we keep the marginal distributions of $X_1$ and $X_2$ the same, but make them correlated in the following way: if $X_1 = i$, then $X_2 = 2^i$. We show that the expected revenue of the optimal auction in this case is exactly $3$. Showing it is at least $3$ is straightforward: consider the auction that posts a price of $1$ on the first stage and a price of $2$ on the second stage. In order to show that the revenue is at most $3$ we apply our duality framework for the case of correlated stages.

At the first stage we push the Myerson type flow (Claim~\ref{clm:e_1 plus m_2}, Section~\ref{sec: single agent upper bound}). Namely, each node $v_1$ pushes all its incoming flows to its immediate successor $\downnode{v}{1} = v_1-1$, i.e. $\lambda_1(v_1, \downnode{v}{1}) = 1-F_1(\downnode{v}{1})$. This flow makes the expected virtual value at stage one equal $\myerson{X_1}=1$ by an argument in Claim~\ref{clm:e_1 plus m_2} in Section~\ref{sec: single agent upper bound}. $\kappa_1(v_1) = 0$, but for correlated stages it can still be the case that a node $(v_1,v_2)$ in the second stage (a node $v_2$ with history/parent $v_1$) can still have positive incoming flow. In fact, for our choice of $\lambda_1$, and the specific correlation of our example, the incoming ``default'' flow to a node $v_2 = 2^{\upnode{v}{1}}$ under parent node $v_1$ is equal to $\lambda_1(\upnode{v}{1},v_1)$, the amount of flow $\upnode{v}{1}$ pushes to $v_1$ in the first stage (See the flow conservation constraints for correlated stages in Appendix~\ref{app: correlated upper bound}). To satisfy the flow conservation constraints we push $\lambda_1(\upnode{v}{1},v_1) = 1-F_1(v_1)$ unit of flow from node $\upnode{v}{2} = 2^{\upnode{v}{1}}$ to node $v_2$, under parent node $v_1$, i.e. $\lambda_2 ( v_1, 2^{\upnode{v}{1}} , 2^{v_1} ) = 1-F_1(v_1) = 1 - F_2(2^{v_1})$. 
Since $f(v_2|v_1) = 0$, unless $v_2 = 2^{v_1}$, the only second stage node that contributes to the objective is the one with the incoming flow (and lower virtual value $\Phi_2$).
\[\mathbb{E}_{v_2}[\Phi_2(v_1, v_2)|v_1] = \Phi_2(v_1, 2^{v_1}) = 2^{v_1} - \frac{1-F(v_1)}{f(v_1)}(2^{v_1+1} - 2^{v_1}) = 2^{v_1} - (2^{v_1+1} - 2^{v_1}) = 0, \]
for all $v_1$, except for $\bar{v}_1 = n+1$ (the highest value in the support), whose corresponding second stage value $\bar{v}_2 = 2^{n+1}$ does not have incoming flows; $\upnode{\bar{v}}{2}$ is not defined. In this case $\Phi_2(\bar{v}_1, \bar{v}_2) = \bar{v}_2$. Thus, the expected virtual value of the second stage comes exclusively from the virtual value of $\bar{v}_2$:

\[ \textstyle \mathbb{E}_{v_1,v_2}[\Phi_2(v_1, v_2)] = f(\bar{v}_1)\mathbb{E}_{v_2}[\Phi_2(\bar{v}_1, v_2)|\bar{v}_2] = f(\bar{v}_1)f(\bar{v}_2|\bar{v}_1)\Phi_2(\bar{v}_1, \bar{v}_2) = 2^{-n}\times2^{n+1} = 2. \]

The primal objective of revenue is upper bounded by $\max_{x\in\mathcal{F},p}\mathcal{L}(\lambda, \kappa, x, p)$, which in our case is at most $\mathbb{E}_{v_1}[\Phi_1(v_1)]+\mathbb{E}_{v_1,v_2}[\Phi_2(v_1, v_2)] = 1+2=3$, since that is the maximum dual value under the flow described in the previous paragraph.


\bibliographystyle{plainnat}
\bibliography{refs}

\begin{appendix}

\section{Revenue Upper Bounds for $n$ Buyers and $\days$ Independent Stages}\label{sec: proof of upper bound}

In this Section we prove Lemma~\ref{lem:upper bound on opt}, using an approach similar to the previous Section.

\subsection{The Partial Lagrangian.}
The optimal dynamic mechanism needs to satisfy the following two types of constraints:

\begin{itemize}
\item Periodic incentive compatibility (PIC). At any stage $k$ for every buyer $i$, assuming truthfulness in all future stages, revealing the true value $v^i_k$ maximizes the buyer's expected utility among all possible values $\hat{v}^i_k$. For stage $k$ and buyer $i$, this constraint can be expressed as: for all $\mathbf{v}_{\leq k-1}$ in $\mathbf{V_{k-1}}$:
\begin{align*}
&\mathbb{E}_{\mathbf{v}^{-i}_k}\left[ v^i_kx^i_k(\mathbf{v}^{-i}_{\leq k},v^{i}_{\leq k})-p^i_k(\mathbf{v}^{-i}_{\leq k},v^{i}_{\leq k}) + \mathbb{E}_{\mathbf{v}^{-i}_{k+1:m}, v^i_{k+1:m}}\left[\sum_{j>k}v^i_jx^i_j(\mathbf{v}^{-i}_{\leq j},v^{i}_{\leq j})- p^i_j(\mathbf{v}^{-i}_{\leq j},v^{i}_{\leq j})\right]\right]  \\
&~~~~\geq \mathbb{E}_{\mathbf{v}^{-i}_k} \left[ 
\vphantom{\mathbb{E}_{\mathbf{v}^{-i}_{k+1:m}, v^i_{k+1:m}}\left[\sum_{j>k}v^i_jx^i_j(\mathbf{v}^{-i}_{\leq j},v^{i}_{< k}, \hat{v}^i_k, v^{i}_{k+1:j})- p^i_j(\mathbf{v}^{-i}_{\leq j},v^{i}_{< k}, \hat{v}^i_k, v^{i}_{k+1:j})\right]}
v^i_kx^i_k(\mathbf{v}^{-i}_{\leq k},v^{i}_{< k}, \hat{v}^i_k)-p^i_k(\mathbf{v}^{-i}_{\leq k},v^{i}_{< k}, \hat{v}^i_k)+ \right. \\
&~~~~~~~~~~~~~~~~~~~~~~~~~\left. \mathbb{E}_{\mathbf{v}^{-i}_{k+1:m}, v^i_{k+1:m}}\left[\sum_{j>k}v^i_jx^i_j(\mathbf{v}^{-i}_{\leq j},v^{i}_{< k}, \hat{v}^i_k, v^{i}_{k+1:j})- p^i_j(\mathbf{v}^{-i}_{\leq j},v^{i}_{< k}, \hat{v}^i_k, v^{i}_{k+1:j})\right] \right]  
\end{align*}

\item Ex-post individual rationality. At any stage $k$, the stage utility of every buyer $i$ is non-negative regardless of the reports from previous stages. For all $\mathbf{v}_{\leq k}$ in $\mathbf{V}_{\leq k}$:
\begin{equation}\label{eq:ex_pos_ir}
v^i_kx^i_k(\mathbf{v}_{\leq k})-p^i_k(\mathbf{v}_{\leq k}) \geq 0 
\end{equation}
As noted in \citet{ashlagi2016}, the following constraints are equivalent to Constrains~\ref{eq:ex_pos_ir}: for all $\mathbf{v}_{\leq k-1}$ in $\mathbf{V}_{k-1}$, and $v^i_k$ in $V_k$ 
\begin{equation}\label{eq:exp_ex_pos_ir}
\mathbb{E}_{\mathbf{v}^{-i}_k}[v^i_kx^i_k(\mathbf{v}_{\leq k})-p^i_k(\mathbf{v}_{\leq k})]  \geq 0
\end{equation}
Any $x^i_k, p^i_k$ satisfying Constraint~\ref{eq:ex_pos_ir} naturally satisfy Constraint~\ref{eq:exp_ex_pos_ir}; on the other hand, for any auction $M$ with $x^i_k, p^i_k$ satisfying \ref{eq:exp_ex_pos_ir}, there exists another auction $\hat{M}$ defined as:
\begin{align*}
\hat{p}^i_k(\mathbf{v}_{\leq k}) = \frac{\mathbb{E}_{\mathbf{v}^{-i}_k}[p^i_k(\mathbf{v}_{\leq k})]}{\mathbb{E}_{\mathbf{v}^{-i}_k}[x^i_k(\mathbf{v}_{\leq k})]} 
&& 
\hat{x}^i_k(\mathbf{v}_{\leq k}) = 
\begin{cases} 1 \quad &\text{if item $k$ is allocated to buyer $i$ in $M$} \\
0 &\text{o.w.}
\end{cases}
\end{align*}
It's easy to verify that $\hat{x}^i_k, \hat{p}^i_k$ satisfy Constraint~\ref{eq:ex_pos_ir} and $\rev{\hat{M},\mathcal{X},n,\days} = \rev{M,\mathcal{X},n,\days}$. We use Constraints~\ref{eq:exp_ex_pos_ir} for the primal program.
\end{itemize}



Thereby, we obtain the following primal program:

\begin{align*}
  &\qquad\textrm{max }  \sum_{k=1}^{\days}\sum_{i=1}^n\sum_{\mathbf{v}_{\leq k}} f(\mathbf{v}_{\leq k})p^i_k(\mathbf{v}_{\leq k})  \notag\\
  &\textrm{subject to: } \notag \\
   &\forall i=1,\dots, n, \forall k=1, \dots \days, \forall \mathbf{v}_{\leq k-1}\in \mathbf{V}_{k-1},v^i_k, \hat{v}^i_k\in V_k: \\
    &\textstyle \sum_{\mathbf{v}^{-i}_k} f(\mathbf{v}^{-i}_k)\left( v^i_kx^i_k(\mathbf{v}^{-i}_{\leq k},v^{i}_{\leq k})-p^i_k(\mathbf{v}^{-i}_{\leq k},v^{i}_{\leq k}) + \sum_{\mathbf{v}_{k+1:\days}}f(\mathbf{v}_{k+1:\days}) \left( \sum_{j>k}v^i_jx^i_j(\mathbf{v}^{-i}_{\leq j},v^{i}_{\leq j}) -p^i_j(\mathbf{v}^{-i}_{\leq j},v^{i}_{\leq j}) \right) \right) \geq \\
   &\textstyle\quad \sum_{\mathbf{v}^{-i}_k} f(\mathbf{v}^{-i}_k)\left( 
   \vphantom{\sum_{\mathbf{v}_{k+1:\days}}f(\mathbf{v}_{k+1:\days}) \left( \sum_{j>k}v^i_jx^i_j(\mathbf{v}^{-i}_{\leq j},v^{i}_{< k}, \hat{v}^i_k, v^{i}_{k+1:j})- p^i_j(\mathbf{v}^{-i}_{\leq j},v^{i}_{< k}, \hat{v}^i_k, v^{i}_{k+1:j}) \right)}
   v^i_kx^i_k(\mathbf{v}^{-i}_{\leq k},v^{i}_{< k}, \hat{v}^i_k)-p^i_k(\mathbf{v}^{-i}_{\leq k},v^{i}_{< k}, \hat{v}^i_k) 
   \right.\\
   &\textstyle\left. ~~~~~~~~~~~~+ \sum_{\mathbf{v}_{k+1:\days}}f(\mathbf{v}_{k+1:\days}) \left( \sum_{j>k}v^i_jx^i_j(\mathbf{v}^{-i}_{\leq j},v^{i}_{< k}, \hat{v}^i_k, v^{i}_{k+1:j})- p^i_j(\mathbf{v}^{-i}_{\leq j},v^{i}_{< k}, \hat{v}^i_k, v^{i}_{k+1:j}) \right) \right) \\
   &\forall k=1,\dots,\days, \forall \mathbf{v}_{\leq k-1}\in \mathbf{V}_{k-1},v^i_k\in V_k: \\
   &\qquad \sum_{\mathbf{v}^{-i}_k} f(\mathbf{v}^{-i}_k)\left( v^i_kx^i_k(\mathbf{v}_{\leq k})-p^i_k(\mathbf{v}_{\leq k}) \right) \geq 0  \\
   & \forall i=1,\dots, n \forall k=1,\dots,\days, \forall \mathbf{v}_{\leq k}\in\mathbf{V}_k: \qquad x^i_k(\mathbf{v}_{\leq k}) \in [0,1] \\
   &\forall k=1,\dots,\days,\forall \mathbf{v}_{\leq k}\in\mathbf{V}_k: \qquad \sum_{i=1}^n x^i_k(\mathbf{v}_{\leq k}) \leq 1
  \end{align*}

We introduce a Lagrangian multiplier $\lambda_k(\mathbf{v}_{<k},v^i_k, \hat{v}^i_k)$ for each periodic IC constraint and $\kappa_k(\mathbf{v}_{<k},v^i_k)$ for each ex-post IR constraint. 
The partial Lagrangian, after re-grouping the terms, is the following:

\begin{align*}
\mathcal{L}(\lambda, \kappa, x, p) =& \sum_{i=1}^n \sum_{k=1}^{\days} \sum_{\mathbf{v}_{\leq k}} p^i_k(\mathbf{v}_{\leq t})\left(f(\mathbf{v}_{\leq k}) - f(\mathbf{v}^{-i}_k)\kappa_k(\mathbf{v}_{< k}, v^i_k) + f(\mathbf{v}^{-i}_k)\sum_{\hat{v}^i_k}(\lambda_k(\mathbf{v}_{< k}, \hat{v}^i_k, v^i_k) \right.\\
&\left.\quad -\lambda_k(\mathbf{v}_{< k}, v^i_k, \hat{v}^i_k)) + \sum_{j=1}^{k-1}\sum_{\hat{v}^i_j}f(\mathbf{v}^{-i}_j, \mathbf{v}_{j+1:k})(\lambda_j(\mathbf{v}_{< j}, \hat{v}^i_j, v^i_j)-\lambda_j(\mathbf{v}_{<j}, v^i_j, \hat{v}^i_j) ) \right) + \\
&\sum_{i=1}^n \sum_{k=1}^{\days} \sum_{\mathbf{v}_{\leq k}} x^i_k(\mathbf{v}_{\leq k})\left(v^i_kf(\mathbf{v}^{-i}_k)\kappa_k(\mathbf{v}_{< k}, v^i_k) + f(\mathbf{v}^{-i}_k)\sum_{\hat{v}^i_k}(v^i_k\lambda_k(\mathbf{v}_{<k}, v^i_k, \hat{v}^i_k) \right.\\
&\left.\quad - \hat{v}^i_k\lambda_k(\mathbf{v}_{<k}, \hat{v}^i_k, v^i_k)) + \sum_{j=1}^{k-1}\sum_{\hat{v}^i_j}f(\mathbf{v}^{-i}_j, \mathbf{v}_{j+1:k})(v^i_k\lambda_j(\mathbf{v}_{< j}, v^i_j, \hat{v}^i_j) - v^i_k\lambda_j(\mathbf{v}_{< j}, \hat{v}^i_j, v^i_j)) \right) \\
\end{align*}

Duality theory tells us for any $\lambda,\kappa\geq 0$, the primal objective is upper bounded by $\max_{x\in\mathcal{F},p}\mathcal{L}(\lambda, \kappa, x, p)$, where $\mathcal{F}$ is the set of possible allocations: 
\begin{equation}\label{eq: nm opt dual upper bound}
\opt{\mathcal{X},n,\days} \leq \max_{x\in\mathcal{F},p}\mathcal{L}(\lambda, \kappa, x, p)
\end{equation}
upper bounds the primal objective. If we want to find non-trivial upper bounds, we need to ensure that $\max_{x\in\mathcal{F},p}\mathcal{L}(\lambda, \kappa, x, p)$ is bounded. This requires that the free variables $p^i_k(\mathbf{v}_{\leq k})$ have multipliers equal to zero. Therefore: for stage $k$, buyer $i$ and reports $\mathbf{v}_{\leq k}$ in $\mathbf{V}_k$
\begin{multline}\label{eq:original flow k}
f(\mathbf{v}_{\leq k}) - f(\mathbf{v}^{-i}_k)\kappa_k(\mathbf{v}_{< k}, v^i_k) + f(\mathbf{v}^{-i}_k)\sum_{\hat{v}^i_k}(\lambda_k(\mathbf{v}_{< k}, \hat{v}^i_k, v^i_k) -\lambda_k(\mathbf{v}_{< k}, v^i_k, \hat{v}^i_k)) \\
+ \sum_{j=1}^{k-1}\sum_{\hat{v}^i_j}f(\mathbf{v}^{-i}_j, \mathbf{v}_{j+1:k})(\lambda_j(\mathbf{v}_{< j}, \hat{v}^i_j, v^i_j)-\lambda_j(\mathbf{v}_{<j}, v^i_j, \hat{v}^i_j) ) = 0 
\end{multline}

Noticing the recursive structure of Equation~\ref{eq:original flow k}, we take the constraint for stage $k-1$, buyer $i$ and reports $\mathbf{v}_{\leq k-1}$ in $\mathbf{V}_{k-1}$:

\begin{multline}\label{eq:flow_k-1}
f(\mathbf{v}_{\leq k-1}) - f(\mathbf{v}^{-i}_{k-1})\kappa_{k-1}(\mathbf{v}_{< k-1}, v^i_{k-1}) + f(\mathbf{v}^{-i}_{k-1})\sum_{\hat{v}^i_{k-1}}(\lambda_{k-1}(\mathbf{v}_{< k-1}, \hat{v}^i_{k-1}, v^i_{k-1}) \\
-\lambda_{k-1}(\mathbf{v}_{< k-1}, v^i_{k-1}, \hat{v}^i_{k-1})) + \sum_{j=1}^{k-2}\sum_{\hat{v}^i_j}f(\mathbf{v}^{-i}_j, \mathbf{v}_{j+1:k-1})(\lambda_j(\mathbf{v}_{< j}, \hat{v}^i_j, v^i_j)-\lambda_j(\mathbf{v}_{<j}, v^i_j, \hat{v}^i_j) ) =0
\end{multline}

Simplify the LHS of Equation~\ref{eq:original flow k} with Equation~\ref{eq:flow_k-1}:
\begin{align*}
0&=f(\mathbf{v}_{\leq k}) - f(\mathbf{v}^{-i}_k)\kappa(\mathbf{v}_{< k}, v^i_k) + f(\mathbf{v}^{-i}_k)\sum_{\hat{v}^i_k}(\lambda_k(\mathbf{v}_{< k}, \hat{v}^i_k, v^i_k) -\lambda_k(\mathbf{v}_{< k}, v^i_k, \hat{v}^i_k)) \\
&\quad + \sum_{j=1}^{k-1}\sum_{\hat{v}^i_j}f(\mathbf{v}^{-i}_j, \mathbf{v}_{j+1:k})(\lambda_j(\mathbf{v}_{< j}, \hat{v}^i_j, v^i_j)-\lambda_j(\mathbf{v}_{<j}, v^i_j, \hat{v}^i_j) ) \\
&=^{\text{Eq.~}\ref{eq:flow_k-1}} f(\mathbf{v}_{\leq k}) - f(\mathbf{v}^{-i}_k)\kappa(\mathbf{v}_{< k}, v^i_k) + f(\mathbf{v}^{-i}_k)\sum_{\hat{v}^i_k}(\lambda_k(\mathbf{v}_{< k}, \hat{v}^i_k, v^i_k) -\lambda_k(\mathbf{v}_{< k}, v^i_k, \hat{v}^i_k))  \\
&\quad + f(\mathbf{v}_k)(f(\mathbf{v}^{-i}_{k-1})\kappa_{k-1}(\mathbf{v}_{< k-1}, v^i_{k-1})  - f(\mathbf{v}_{\leq k-1})) \\
&= f(\mathbf{v}^{-i}_{k-1}, \mathbf{v}_k)\kappa_{k-1}(\mathbf{v}_{<k-1},v^i_{k-1})- f(\mathbf{v}^{-i}_k)\kappa_k(\mathbf{v}_{< k}, v^i_k)  \\
&\quad  + f(\mathbf{v}^{-i}_k)\sum_{\hat{v}^i_k}(\lambda_k(\mathbf{v}_{< k}, \hat{v}^i_k, v^i_k) -\lambda_k(\mathbf{v}_{< k}, v^i_k, \hat{v}^i_k)) = 0
\end{align*}

Since the last line equals zero, we obtain:
\begin{multline}\label{eq:flow_k}
f(\mathbf{v}^{-i}_{k-1}, \mathbf{v}_k)\kappa_{k-1}(\mathbf{v}_{<k-1},v^i_{k-1}) - f(\mathbf{v}^{-i}_k)\kappa_k(\mathbf{v}_{< k}, v^i_k)= \\
f(\mathbf{v}^{-i}_k)\sum_{\hat{v}^i_k}(\lambda_k(\mathbf{v}_{< k}, v^i_k, \hat{v}^i_k) - \lambda_k(\mathbf{v}_{< k}, \hat{v}^i_k, v^i_k)) 
\end{multline}

Recall that we call solutions of $\lambda,\kappa\geq 0$ that satisfy Equation~\ref{eq:original flow k} \textbf{useful} solutions. Then, given a set of useful $\lambda,\kappa$ we can simplify $\mathcal{L}(\lambda,\kappa,p,x)$ using Equations~\ref{eq:original flow k},~\ref{eq:flow_k-1} and~\ref{eq:flow_k}.
Gather all the terms in $\mathcal{L}(\lambda, \kappa, x, p)$ that involve $x^i_k$ for $k>1$ and simplify the terms using Equation~\ref{eq:flow_k}:
\begin{align*}
&\sum_{i=1}^n\sum_{k=2}^{\days}\sum_{\mathbf{v}_{\leq k}} x^i_k(\mathbf{v}_{\leq k})\left(v^i_kf(\mathbf{v}^{-i}_k)\kappa_k(\mathbf{v}_{< k}, v^i_k) + f(\mathbf{v}^{-i}_k)\sum_{\hat{v}^i_k}(v^i_k\lambda_k(\mathbf{v}_{<k}, v^i_k, \hat{v}^i_k) - \hat{v}^i_k\lambda_k(\mathbf{v}_{<k}, \hat{v}^i_k, v^i_k))\right. \\
&\left.\quad ~~~~~~~~~~~~~~~~~~~~~~~~~~ + v^i_k\sum_{j=1}^{k-1}\sum_{\hat{v}^i_j}f(\mathbf{v}^{-i}_j, \mathbf{v}_{j+1:k})(\lambda_j(\mathbf{v}_{< j}, v^i_j, \hat{v}^i_j) - \lambda_j(\mathbf{v}_{< j}, \hat{v}^i_j, v^i_j)) \right) \\
&=^{\text{Eq.~}\ref{eq:flow_k}} \sum_{i=1}^n\sum_{k=2}^{\days}\sum_{\mathbf{v}_{\leq k}} x^i_k( \mathbf{v}_{\leq k}) \left( v^i_kf(\mathbf{v}^{-i}_k)\kappa_k(\mathbf{v}_{< k}, v^i_k) + f(\mathbf{v}^{-i}_k)\sum_{\hat{v}^i_k}(v^i_k\lambda_k(\mathbf{v}_{<k}, v^i_k, \hat{v}^i_k) - \hat{v}^i_k\lambda_k(\mathbf{v}_{<k}, \hat{v}^i_k, v^i_k)) \right. \\
&\left.\quad ~~~~~~~~~~~~~~~~~~~~~~~~~~
\vphantom{v^i_kf(\mathbf{v}^{-i}_k)\kappa_k(\mathbf{v}_{< k}, v^i_k) + f(\mathbf{v}^{-i}_k)\sum_{\hat{v}^i_k}(v^i_k\lambda_k(\mathbf{v}_{<k}, v^i_k, \hat{v}^i_k) - \hat{v}^i_k\lambda_k(\mathbf{v}_{<k}, \hat{v}^i_k, v^i_k))}
- v^i_kf(\mathbf{v}_k)(f(\mathbf{v}^{-i}_{k-1})\kappa_{k-1}(\mathbf{v}_{< k-1}, v^i_{k-1})  - f(\mathbf{v}_{\leq k-1})) \right) \\
&= \sum_{i=1}^n\sum_{k=2}^{\days}\sum_{\mathbf{v}_{\leq k}} x^i_k(\mathbf{v}_{\leq k})\left(v^i_kf(\mathbf{v}_{\leq k-1})f(\mathbf{v}_k)- f(\mathbf{v}^{-i}_k)\sum_{\hat{v}^i_k}(\hat{v}^i_k - v^i_k)\lambda_k(\mathbf{v}_{<k},\hat{v}^i_k, v^i_k)\right) \\
&= \sum_{i=1}^n\sum_{k=2}^{\days}\sum_{\mathbf{v}_{\leq k}} x^i_kf(\mathbf{v}_{\leq k})(\mathbf{v}_{\leq k})\left(v^i_k- \frac{1}{f(\mathbf{v}_{\leq k-1},v^i_k)}\sum_{\hat{v}^i_k}(\hat{v}^i_k - v^i_k)\lambda_k(\mathbf{v}_{<k},\hat{v}^i_k, v^i_k)\right) \\
&= \sum_{i=1}^n\sum_{k=2}^{\days}\sum_{\mathbf{v}_{\leq k}} x^i_kf(\mathbf{v}_{\leq k})\Phi_k(\mathbf{v}_{\leq k})
\end{align*}

where $\Phi^i_k(\mathbf{v}_{\leq k}) = v^i_k- \frac{1}{f(\mathbf{v}_{\leq k-1},v^i_k)}\sum_{\hat{v}^i_k}(\hat{v}^i_k - v^i_k)\lambda_k(\mathbf{v}_{<k},\hat{v}^i_k, v^i_k)$.
Finally, simplify the terms involving $x^i_1(\mathbf{v}_1)$ using Equation~\ref{eq:original flow k} evaluated at $k = 1$:

\begin{align*}
&\sum_{i=1}^n\sum_{\mathbf{v}_1} x^i_1(\mathbf{v}_1)\left(v^i_1f(\mathbf{v}^{-i}_1)\kappa_1(v^i_1) + f(\mathbf{v}^{-i}_1)\sum_{\hat{v}^i_1}(v^i_1\lambda_1(v^i_1,\hat{v}^i_1) -\hat{v}^i_1\lambda_1(\hat{v}^i_1,v^i_1))\right) \\
&=^{\text{Eq.~}\ref{eq:original flow k}}\sum_{i=1}^n\sum_{\mathbf{v}_1} x^i_1(\mathbf{v}_1)\left(v^i_1f(\mathbf{v}_1) - f(\mathbf{v}^{-i}_1)\sum_{\hat{v}^i_1}(\hat{v}^i_1 - v^i_1)\lambda_1(\hat{v}^i_1,v^i_1)\right)\\
&= \sum_{i=1}^n\sum_{\mathbf{v}_1} x^i_1(\mathbf{v}_1)f(\mathbf{v}_1)\left(v^i_1 - \frac{1}{f(v^i_1)}\sum_{\hat{v}^i_1}(\hat{v}^i_1 - v^i_1)\lambda_1(\hat{v}^i_1,v^i_1)\right) \\
&= \sum_{i=1}^n\sum_{\mathbf{v}_1} x^i_1(\mathbf{v}_1)f(\mathbf{v}_1)\Phi_1(\mathbf{v}_1)
\end{align*}

where $\Phi^i_1(\mathbf{v}_1) = v^i_1 - \frac{1}{f(v^i_1)}\sum_{\hat{v}^i_1}(\hat{v}^i_1 - v^i_1)\lambda_1(\hat{v}^i_1,v^i_1)$.

Combining all the observations so far, we have that given any useful solution $\lambda,\kappa$:
\begin{align*}
\mathcal{L}(\lambda, \kappa, x, p)=& \sum_{i=1}^n \sum_{k=1}^{\days} \sum_{\mathbf{v}_{\leq k}} x^i_k(\mathbf{v}_{\leq k})f(\mathbf{v}_{\leq k})\Phi^i_k(\mathbf{v}_{\leq k}) \\
\end{align*}

\subsection{Main Claim}

\begin{claim}\label{clm:m_j plus e_not_j}
For $n$ agents, $\days$ independent stages, ex-post IR and PIC dynamic mechanisms
\[ \opt{\mathcal{X},n,m} \leq \myerson{\mathbf{X}_j}  + \sum_{k=1, k \neq j}^{\days} \expectation{( X_{k})_{1:n}} \]
for any $j = 1,\dots, \days$.
\end{claim}

\begin{proof}
Consider the following dual solution:
\begin{align*}\kappa_k (\mathbf{v}_{<k}, v^i_k) = 
\begin{cases}
f(\mathbf{v}_{<k},v^i_k) & k< j \\
\frac{f(\mathbf{v}_{<k},v^i_k)}{f(v^i_j)} & k\geq j \land v^i_j = \underline{v}^i_j\\
0 & \text{o.w.} 
\end{cases} 
&&
\lambda_k (\mathbf{v}_{<k}, v^i_k,\hat{v}^i_k) = 
\begin{cases}
f(\mathbf{v}_{<k})(1 - F(\hat{v}^i_k)) & k=j\land \hat{v}^i_k = {\downnode{v^i}{k}}^{\text{Dfn \ref{dfn: Successor and Predecessor}}}\\
0 &\text{o.w.} \\
\end{cases}
\end{align*}

It's easy to verify that Constraint~\ref{eq:original flow k} is satisfied. These flows induce virtual value $\Phi^i_k(\mathbf{v}_k) = v^i_k$ for all $k\neq j$. For stage $j$, $\Phi^i_j(\mathbf{v}_{\leq j})$ becomes $\phi(v^i_j)$, Myerson's virtual value for $X_j$. For simplicity we assume that $X_j$ is regular, so that the virtual values induced are non-decreasing; if this is not the case, we use an ``ironing'' procedure to the flow $\lambda_j$ in stage $j$, similar to~\cite{cai2016}.

Then by Inequality~\ref{eq: nm opt dual upper bound}:
\begin{align*}
\opt{\mathcal{X},n,m} &\leq \max_{x,p}\mathcal{L}(\lambda, \kappa, x, p) \\
&= \max_{x,p}\sum_{i=1}^n \sum_{k=1}^{\days} \sum_{\mathbf{v}_{\leq k}} f(\mathbf{v}_{\leq k})x^i_k(\mathbf{v}_{\leq k})\Phi^i_k(\mathbf{v}_{\leq k}) \\
&=^{\text{Dnf \ref{dfn: Successor and Predecessor}}} \max_{x,p}\sum_{i=1}^n \sum_{\mathbf{v}_{\leq j}} f(\mathbf{v}_{\leq j})x^i_j(\mathbf{v}_{\leq j}) \left( v^i_j- \frac{1}{f(v^i_j, \mathbf{v}_{< j})}(\upnode{v^i}{j} - v^i_j)f(\mathbf{v}_{<j})(1 - F(v^i_j)) \right) \\
&\qquad ~~~~~~~~~~~~~ + \sum_{i=1}^n \sum_{k=1,k\neq j}^{\days} \sum_{\mathbf{v}_{\leq k}} f(\mathbf{v}_{\leq k})v^i_k \\
&=^{\text{Dnf \ref{dfn: Myerson's virtual value}}} \max_{x,p}\sum_{\mathbf{v}_{< j}} f(\mathbf{v}_{<j})\sum_{\mathbf{v}_j}f(\mathbf{v}_j)\sum_{i=1}^nx^i_j(\mathbf{v}_{\leq j})\phi(v^i_j) + \sum_{k=1, k \neq j}^{\days} \expectation{( X_{k})_{1:n}}  \\
&= \sum_{\mathbf{v}_{< j}} f(\mathbf{v}_{<j})\max_{x,p} \sum_{\mathbf{v}_j}f(\mathbf{v}_j)x_j(\mathbf{v}_{\leq j})\phi(\mathbf{v}_j) + \sum_{k=1, k \neq j}^{\days} \expectation{( X_{k})_{1:n}} \\
&=\myerson{\mathbf{X}_j}  + \sum_{k=1, k \neq j}^{\days} \expectation{( X_{k})_{1:n}} \qedhere
\end{align*}
\end{proof}


\section{Correlated stages with a stochastic dominance structure}\label{app: correlated upper bound}

We prove a simliar result to Lemma~\ref{lem:upper bound on opt} under a certain stochastic dominance condition.

\begin{lemma}\label{lem: upper bound on opt corr}
For one agent, $\days$ correlated stages, ex-post IR and PIC dynamic mechanisms
\[\opt{\mathcal{X}, 1, \days} \leq \min_{j=1,\dots,m}\left\{ \expectation{\myerson{X_j| X_{< j}}} +\sum_{k=1,k\neq j}^{\days} \expectation{X_k} \right\}\]
if for any stage $k$, and any two histories $v_{<k}$, $v'_{<k}$ that satisfy $\forall t<k, v_t \geq v'_t$, the conditional distribution $X_k | v_{<k}$ stochastically dominates the conditional distribution $X_k | v'_{<k}$.
\end{lemma}

The proof uses the duality framework. First we construct the revenue maximizing LP, with the objective being the expected revenue and the constraints enforcing PIC and ex-post IR. We get a partial Lagrangian $\mathcal{L}(\lambda, \kappa, x, p)$, and by duality:
\[\opt{\mathcal{X},1,m} \leq \max_{x,p}\mathcal{L}(\lambda,\kappa, x, p) \quad \text{ where }\lambda, \kappa\geq 0. \]

Then for each stage $j$ we provide a feasible assignment of $\lambda, \kappa$ such that
\[\max_{x,p}\mathcal{L}(\lambda,\kappa, x, p)\leq \expectation{\myerson{X_j|X_{<j}}} + \sum_{k=1, k \neq j}^{\days} \expectation{ X_{k} }.\]

\subsubsection*{The Primal.}

\begin{align*}
  &\textrm{max } \sum_{k=1}^{\days}\sum_{v_{\leq k}}f(v_{\leq k})p_k(v_{\leq k})  \notag\\
  &\textrm{subject to: } \notag \\
   &\forall k=1,\dots,\days,  v_{< k},v_k,\hat{v}_k:  \\
    &~ v_k x_k(v_{< k},v_k) - p_k(v_{< k},v_k) + \sum_{v_{k+1:\days}}f(v_{k+1:\days}|v_{\leq k})\sum_{j >k}^{\days} \left( v_jx_j(v_{\leq k}, v_{k+1:j}) - p_j(v_{\leq k}, v_{k+1:j}) \right) \geq \\
    &~~v_kx_k(v_{< k},\hat{v}_k) - p_k(v_{< k},\hat{v}_k) + \sum_{v_{k+1:\days}} f(v_{k+1:\days}|v_{\leq k})\sum_{j>k}^{\days} \left( v_jx_j(v_{< k-1}, \hat{v}_k, v_{k+1:j}) - p_t(v_{< k-1}, \hat{v}_k, v_{k+1:j}) \right) \\
   &\forall k=1,\dots,\days,  v_{<k}, v_k,\hat{v}_k :\\
    &~~~~~~ v_kx_k(v_{\leq k}) - p_k(v_{\leq k}) \geq 0
  \end{align*}

\subsubsection*{The partial Lagrangian.}

Introduce a lagrangian multiplier $\lambda_k(v_{<k},v_k, \hat{v}_k)$ for each periodic IC constraint and $\kappa_k(v_{<k},v_k)$ for each ex-post IR constraint. Then lagrangify the primal program to obtain the following partial dual,

\begin{align*}
\mathcal{L}(\lambda, \kappa, x, p) =& \sum_{k=1}^{\days} \sum_{v_{\leq k}} p_k(v_{\leq k}) \left( f(v_{\leq k}) -\kappa_k(v_\leq k)   + \sum_{\hat{v}_k}(\lambda_k(v_{<k},\hat{v}_k, v_k)- \lambda_k(v_{\leq k},\hat{v}_k)) \right.\\
&\left. \quad +  \sum_{j=1}^{k-1}\sum_{\hat{v}_j} (f(v_{j+1:k}|v_{< j}\hat{v}_j)\lambda_j(v_{<j}\hat{v}_j,v_j) - f(v_{j+1:k}|v_{\leq j})\lambda_j(v_{\leq j},\hat{v}_j) ) \right) + \\
&\sum_{k=1}^{\days} \sum_{v_{\leq k}} x_k(v_{\leq k}) \left( v_k\kappa_k(v_{< k}, v_k) + \sum_{\hat{v}_k}(v_k\lambda_k(v_{<k}, v_k, \hat{v}_k)- \hat{v}_k\lambda_k(v_{<k}, \hat{v}_k, v_k)) \right. \\
&\left.\quad  + \sum_{j=1}^{k-1}\sum_{\hat{v}_j}v_k(f(v_{j+1:k}|v_{\leq j})\lambda_j(v_{\leq j},\hat{v}_j) - f(v_{j+1:k}|v_{< j}\hat{v}_j)\lambda_j(v_{<j},\hat{v}_j,v_j)) \right) \\
\end{align*}

Since $p_k(v_{\leq k})$ is an unconstrained variable, for the upper bound $\max_{x,p}\mathcal{L}(\lambda,\kappa, x, p)$ to be bounded, the multipliers of $p_k(v_{\leq k})$ must equal to zero. As a result, $\lambda$ and $\kappa$ form a ``flow'' that obeys the following flow conservation constraints: $\quad\forall k=1, \dots \days, v_{\leq k}$:

\begin{multline}\label{eq: corr flow k}
f(v_{\leq k}) -\kappa_k(v_\leq k)  + \sum_{\hat{v}_j}(\lambda_k(v_{<k}\hat{v}_k, v_k)- \lambda_k(v_{\leq k},\hat{v}_k)) \\
+\sum_{j=1}^{k-1}\sum_{\hat{v}_j} (f(v_{j+1:k}|v_{< j}\hat{v}_j)\lambda_j(v_{<j}\hat{v}_j,v_j) - f(v_{j+1:k}|v_{\leq j})\lambda_j(v_{\leq j},\hat{v}_j) )) = 0 
\end{multline}

The flow conservation constraints for $p_{k-1}(v_{\leq k-1})$ give:

\begin{multline}\label{eq:corr_flow_k-1}
- \sum_{\hat{v}_{k-1}}\lambda_{k-1}(v_{\leq k-1},\hat{v}_{k-1}) -  \sum_{j=1}^{k-2}\sum_{\hat{v}_j} f(v_{j+1:k-1}|v_{\leq j})\lambda_j(v_{\leq j},\hat{v}_j)\\
\quad= -\sum_{\hat{v}_{k-1}}\lambda_{k-1}(v_{<k-1}\hat{v}_{k-1}, v_{k-1})-\sum_{j=1}^{k-2}\sum_{\hat{v}_j}  f(v_{j+1:k-1}|v_{< j}\hat{v}_j)\lambda_j(v_{<j}\hat{v}_j,v_j) +\kappa_{k-1}(v_{\leq k-1}) - f(v_{\leq k-1})
\end{multline}

Note the recursive structure in the mutlipliers of $p_k(v_{\leq k})$. In $\mathcal{L}(\lambda,\kappa,x,p)$, in the terms containing $p_k(v_{\leq k})$, we can substitute the LHS of \ref{eq:corr_flow_k-1} with the RHS of \ref{eq:corr_flow_k-1}:

\begin{align*}
&p_k(v_{\leq k}) \left( f(v_{\leq k}) -\kappa_k(v_\leq k) + \sum_{\hat{v}_k} \left( \lambda_k(v_{<k}\hat{v}_k, v_k)- \lambda_k(v_{\leq k},\hat{v}_k) \right) \right.\\
&\left.\quad + \sum_{j=1}^{k-1}\sum_{\hat{v}_j} f(v_{j+1:k}|v_{< j}\hat{v}_j)\lambda_j(v_{<j}\hat{v}_j,v_j) - \sum_{j=1}^{k-1}\sum_{\hat{v}_j}f(v_{j+1:k}|v_{\leq j})\lambda_j(v_{\leq j},\hat{v}_j) \right) \\
&= p_k(v_{\leq k})\left( f(v_{\leq k}) -\kappa_k(v_\leq k) + \sum_{\hat{v}_k}(\lambda_k(v_{<k}\hat{v}_k, v_k)- \lambda_k(v_{\leq k},\hat{v}_k))+ \sum_{j=1}^{k-1}\sum_{\hat{v}_j} f(v_{j+1:k}|v_{< j}\hat{v}_j)\lambda_j(v_{<j}\hat{v}_j,v_j)\right.\\
&\left.\quad + f(v_{j+1:k}|v_{\leq k-1})(-\sum_{\hat{v}_{k-1}}\lambda_{k-1}(v_{\leq k-1},\hat{v}_{k-1}) - \sum_{j=1}^{k-2}\sum_{\hat{v}_j} f(v_{j+1:k-1}|v_{\leq j})\lambda_j(v_{\leq j},\hat{v}_j))\right) \\
&= p_k(v_{\leq k})\left( f(v_{\leq k}) -\kappa_k(v_\leq k)  + \sum_{\hat{v}_k}(\lambda_k(v_{<k}\hat{v}_k, v_k)- \lambda_k(v_{\leq k},\hat{v}_k))+ \sum_{j=1}^{k-1}\sum_{\hat{v}_j} f(v_{j+1:k}|v_{< j}\hat{v}_j)\lambda_j(v_{<j}\hat{v}_j,v_j)\right.\\
&\left.\quad +f(v_k|v_{\leq k-1})\kappa_{k-1}(v_{\leq k-1}) - f(v_k|v_{\leq k-1})f(v_{\leq k-1})  - f(v_k|v_{\leq k-1})\sum_{\hat{v}_{k-1}}\lambda_{k-1}(v_{<k-1}\hat{v}_{k-1}, v_{k-1})\right.\\
&\left.\quad -f(v_k|v_{\leq k-1})\sum_{j=1}^{k-2}\sum_{\hat{v}_j}  f(v_{j+1:k-1}|v_{< j}\hat{v}_j)\lambda_j(v_{<j}\hat{v}_j,v_j) \right) \\
&= p_k(v_{\leq k}) \left( f(v_k|v_{\leq k-1})\kappa_{k-1}(v_{\leq k-1}) -\kappa_k(v_\leq k) + \sum_{\hat{v}_k}(\lambda_k(v_{<k}\hat{v}_k, v_k)- \lambda_k(v_{\leq k},\hat{v}_k)) \right.\\
&\left.\quad + \sum_{j=1}^{k-1}\sum_{\hat{v}_j} (f(v_k|v_{<j},\hat{v}_j,v_{j+1:k-1}) - f(v_k|v_{\leq k-1}))f(v_{j+1:k-1}|v_{< j}\hat{v}_j) \lambda_j(v_{<j}\hat{v}_j,v_j) \right)
\end{align*}  

Then, by the conservation constraints of $p_k(v_{\leq k})$,

\begin{multline}\label{eq:corr_flow_k}
\textstyle\sum_{j=1}^{k-1}\sum_{\hat{v}_j} (f(v_{j+1:k}|v_{\leq j})\lambda_j(v_{\leq j},\hat{v}_j) - f(v_{j+1:k}|v_{< j}\hat{v}_j)\lambda_j(v_{<j},\hat{v}_j,v_j)) \\
\textstyle = f(v_{\leq k}) -\kappa_k(v_\leq k) + \sum_{\hat{v}_k}(\lambda_k(v_{<k},\hat{v}_k, v_k)- \lambda_k(v_{\leq k},\hat{v}_k))
\end{multline}

Similarly, in $\mathcal{L}(\lambda,\kappa,x,p)$ in the terms containing $x_k(v_{\leq k})$, substitute the LHS of \ref{eq:corr_flow_k} with the RHS of \ref{eq:corr_flow_k} , we obtain

\begin{align*}
&x_k(v_{\leq k})\left( v_k\kappa_k(v_{< k}, v_k) + \sum_{\hat{v}_k}(v_k\lambda_k(v_{<k}, v_k, \hat{v}_k)- \hat{v}_k\lambda_k(v_{<k}, \hat{v}_k, v_k))\right. \\
&\left.\quad  + v_k\sum_{j=1}^{k-1}\sum_{\hat{v}_j}(f(v_{j+1:k}|v_{\leq j})\lambda_j(v_{\leq j},\hat{v}_j) - f(v_{j+1:k}|v_{< j}\hat{v}_j)\lambda_j(v_{<j},\hat{v}_j,v_j)) \right) \\
&= x_k(v_{\leq k})\left( v_k\kappa_k(v_{< k}, v_k) + \sum_{\hat{v}_k}(v_k\lambda_k(v_{<k}, v_k, \hat{v}_k)- \hat{v}_k\lambda_k(v_{<k}, \hat{v}_k, v_k)) \right. \\
&\left.\quad + v_k(f(v_{\leq k}) -\kappa_k(v_\leq k) + \sum_{\hat{v}_k}(\lambda_k(v_{<k},\hat{v}_k, v_k)- \lambda_k(v_{\leq k},\hat{v}_k)) \right) \\
&= x_k(v_{\leq k})f(v_{\leq k}) \left( v_k - \frac{1}{f(v_{\leq k})}\sum_{\hat{v}_k}(\hat{v}_k - v_k) \lambda_k(v_{<k},\hat{v}_k, v_k) \right)
\end{align*}

Now $\mathcal{L}(\lambda,\kappa,x,p)$ can be simplified as:

\[
\textstyle \mathcal{L}(\lambda,\kappa,x,p) = \sum_{k=1}^{\days} \sum_{v_{\leq k}} x_k(v_{\leq k})\left( v_kf(v_{\leq k}) - \sum_{\hat{v}_k}(\hat{v}_k - v_k) \lambda_k(v_{<k}\hat{v}_k, v_k) \right)
\]

For any stage $j$, the following assignment satisfies all the flow conservation constraints:
\[\kappa_k (v_{<k}, v_k) = 
\begin{cases}
k\geq j \land \forall t = j,\dots, k-1, v_t = \underline{v}_t: &f(v_{<j})f(v_k|v_{<k}) \\
\text{o.w.}: & 0
\end{cases}\]

\[\lambda_k (v_{<k}, v_k,\hat{v}_k) = 
\begin{cases}
k< j: & 0 \\
k=j\land \hat{v}_k = \downnode{v}{k}: & f(v_{<k})(1 - F(\hat{v}_k|v_{<k}))\\
k>j \land \hat{v}_k = \downnode{v}{k} \\
\land\forall t = j,\dots, k-1, v_t = \underline{v}_t: & \sum\limits_{v'_k \geq v_k}\sum\limits_{t=j}^{k-1}(f(v'_k|v_{<t},\upnode{v}{t},v_{t+1:k-1})- f(v'_k|v_{\leq k-1}))\times \\
&\,f(v_{t+1:k-1}|v_{< t},\upnode{v}{t})\lambda_t(v_{<t},\upnode{v}{t},v_t)+ \\
&\, f(v_{<j})(1-F(\hat{v}_k|v_{<k}))\\
k>j \land \hat{v}_k = \downnode{v}{k} \\
\land\exists t = j,\dots, k-1, v_t \neq \underline{v}_t: &\sum\limits_{v'_k \geq v_k}\sum\limits_{t=j}^{k-1}(f(v'_k|v_{<t},\upnode{v}{t},v_{t+1:k-1})- f(v'_k|v_{\leq k-1}))\times \\
&\,f(v_{t+1:k-1}|v_{< t},\upnode{v}{t})\lambda_t(v_{<t},\upnode{v}{t},v_t) \\
\text{o.w.} : &f(v_{\leq k})  \\
\end{cases}\]

It's easy to see that all $\kappa_k$ and all $\lambda_k$ that fall into the first or last cases are non-negative. We need to verify that the flows $\lambda_k (v_{<k}, v_k,\downnode{v}{k})$ for $k>j$ have non-negative values:
\begin{align*}
\lambda_k (v_{<k}, v_k,\downnode{v}{k}) &\geq \sum\limits_{v'_k \geq v_k}\sum\limits_{t=j}^{k-1} \left( f(v'_k|v_{<t},\upnode{v}{t},v_{t+1:k-1})- f(v'_k|v_{\leq k-1}) \right) f(v_{t+1:k-1}|v_{< t},\upnode{v}{t})\lambda_t(v_{<t},\upnode{v}{t},v_t) \\
&= \sum\limits_{t=j}^{k-1}f(v_{t+1:k-1}|v_{< t},\upnode{v}{t})\lambda_t(v_{<t},\upnode{v}{t},v_t) \times \sum\limits_{v'_k \geq v_k}(f(v'_k|v_{<t},\upnode{v}{t},v_{t+1:k-1})- f(v'_k|v_{\leq k-1})) \\
&= \sum\limits_{t=j}^{k-1}f(v_{t+1:k-1}|v_{< t},\upnode{v}{t})\lambda_t(v_{<t},\upnode{v}{t},v_t) \left( F(\downnode{v}{k}|v_{<k}) - F( \downnode{v}{k} |v_{<t},\upnode{v}{t},v_{t+1:k-1})\right).
\end{align*}

By the stochastic dominance structure, since the two histories only differ in stage $t$, and $v_t < \upnode{v}{t}$, $X_{k}| v_{<t},\upnode{v}{t},v_{t+1:k-1}$ stochastically dominates $X_{k}| v_{<k}$. As a result,
\[F( \downnode{v}{k} | v_{<k} )- F( \downnode{v}{k} |v_{<t},\upnode{v}{t},v_{t+1:k-1}) \geq 0, \]
and therefore $\lambda_k (v_{<k}, v_k,\downnode{v}{k})\geq 0$.
Plugging this assignment of $\lambda$ and $\kappa$ to the dual expression gives:

\begin{align*}
\opt{\mathcal{X}, 1, \days} \leq& \max_{x,p} \mathcal{L}(\lambda, \kappa, x,p) \\
=& \max_{x,p}  \sum_{k=1}^{\days} \sum_{v_{\leq k}} p_k(v_{\leq k})\times 0 + \sum_{k=1}^{\days} \sum_{v_{\leq k}} f(v_{\leq k})x_k(v_{\leq k})\left( v_k - \frac{1}{f(v_{\leq k})}\sum_{\hat{v}_k}(\hat{v}_k - v_k) \lambda_k(v_{<k}\hat{v}_k, v_k) \right) \\
=&  \sum_{k=1}^{j-1} \sum_{v_{\leq k}} f(v_{\leq k})x_k(v_{\leq k})v_k + \\
& \sum_{v_{\leq j}}f(v_{\leq j})x^*_j(v_{\leq j})[v_j - \frac{1}{f(v_{\leq j})}(\upnode{v}{j} - v_j)f(v_{<j})(1 - F(v_j|v_{<j}))] + \\
& \sum_{k=j+1}^{\days} \sum_{v_{\leq k}}f(v_{\leq k})x^*_k(v_{\leq k})\left( v_k - \frac{1}{f(v_{\leq k})}(\upnode{v}{k} - v_k) \lambda_k(v_{<k},\upnode{v}{k}, v_k) \right) \\
\leq &\max_{x,p} \sum_{v_{\leq j}}f(v_{\leq j})x_j(v_{\leq j}) \left( v_j - \frac{1 - F(v_j|v_{<j})}{f(v_j|v_{< j})}(\upnode{v}{j} - v_j) \right) + \sum_{k=1}^{j-1}\expectation{X_k} + \sum_{k=j+1}^{\days} \expectation{X_k} \\
= &\max_{x,p} \sum_{v_{<j}}f(v_{<j})\sum_{v_j}f(v_j|v_{<j})x_j(v_{\leq j})\left( v_j - \frac{1 - F(v_j|v_{<j})}{f(v_j|v_{< j})}(\upnode{v}{j} - v_j) \right) + \sum_{k=1,k\neq j}^{\days} \expectation{X_k} \\
=& \expectation{\myerson{X_j|X_{<j}}} + \sum_{k=1,k\neq j}^{\days} \expectation{X_k}
\end{align*}

This upper bound holds for all $j=1,\dots,\days$; Lemma~\ref{lem: upper bound on opt corr} is implied.
\section{Lower Bounds on The Competition Complexity}\label{sec: lower bounds on cc}

\begin{lemma}\label{lemma: lower bound on cc for independent }
For independent stages, $\days-1$ MHR and $1$ regular stage, and ex-post IR auctions, the Competition Complexity is at least $(e-1)n$, even for auctions that are incentive compatible in a perfect Bayesian equilibrium.
\end{lemma}

\begin{proof}
For the first $\days-1$ stages, let the value distributions be $X = Exp(1,V)$ is an exponential distribution with parameter $\lambda=1$, truncated at some large value $V$\footnote{Truncated exponential distributions have monotone hazard rate.}. The value distribution $Y$ for the last stage is an equal revenue distribution truncated at some large value $\hat{V}$\footnote{ $Y$ has revenue $1$ for a single agent, and expectation $\ln \left( \hat{V} \right)$.}, with $F_y(x) = 1-\frac{1}{x}$.
The following auction is IC in a perfect Bayesian equilibrium and ex-post IR:
\begin{itemize}
\item In the first $\days-1$ stages run a first price auction: the winner is the buyer with the highest value with a payment equal to that value.
\item In the last stage, the item is given for free to buyer $i$, with probability equal to $\frac{\sum_{k=1}^{\days-1} p^i_k }{\expectation{Y} }$, where $p^i_k$ is the price payed by agent $i$ in stage $k$. In other words, every buyer $i$ wins by bidding $v$, the probability that she gets the last stage item is increased by $\frac{v}{ \expectation{Y} }$.
\end{itemize}

The revenue of this auction is 
\begin{equation}\label{eq: revenue of auction}
(\days-1) \cdot \expectation{ X_{1:n} }.
\end{equation}

\begin{claim}\label{claim: equal revenue order statistics}
The revenue of VCG at every stage with $c$ additional buyers is at most
\begin{equation}\label{eq: revenue of vcg}
(\days-1) \cdot \expectation{ X_{2:n+c} } + \expectation{ Y_{2:n+c} } \leq (\days-1) \cdot \expectation{ X_{2:n+c} } + n + c.
\end{equation}
\end{claim}

\begin{proof}
For simplicity we show the proof for an untruncated equal revenue distribution $Y$. $F_{2:n+c}(y) = F^{n+c}(y) + (n+c)F^{n+c-1}(y)\left( 1 - F(y) \right) = \left( 1 - \frac{1}{y} \right)^{n+c} + (n+c) \left( 1 - \frac{1}{y} \right)^{n+c-1} \frac{1}{y} = \left( \frac{y-1}{y} \right)^{n+c} \left( 1 - \frac{n+c}{1-y}\right)$. Therefore, $\expectation{Y_{2:n+c}} = \int_1^{\infty} 1 - \left( \frac{y-1}{y} \right)^{n+c} \left( 1 - \frac{n+c}{1-y}\right) dy$. The antiderivative of $1 - \left( \frac{y-1}{y} \right)^{n+c} \left( 1 - \frac{n+c}{1-y}\right)$ is $y\left( 1 - \left( \frac{y-1}{y} \right)^{n+c} \right)$. Therefore:
\begin{align*}
\expectation{Y_{2:n+c}} &= \int_1^{\infty} 1 - \left( \frac{y-1}{y} \right)^{n+c} \left( 1 - \frac{n+c}{1-y}\right) dy \\
&= \left[ y\left( 1 - \left( \frac{y-1}{y} \right)^{n+c} \right) \right]_1^{\infty} \\
&= \left( \lim_{y \rightarrow \infty} \frac{y^{n+c} - (y-1)^{n+c} }{y^{n+c-1} }\right) - 1 \\
&= n+c - 1
\end{align*}
So \ref{eq: revenue of vcg} holds.
\end{proof}

We want to find the smallest $c$ such that right-hand side of ~\ref{eq: revenue of vcg} is at least~\ref{eq: revenue of auction}. For simplicity we compute the expected order statistics of $X$ as if it is an untruncated exponential distribution. We can pick $V$ large enough such that the conclusion is the same.

For an exponential distribution $Z \sim Exp(1)$ we have that $\expectation{Z_{1:n}} = \sum_{i=1}^n \frac{1}{i} = H_n$ and $\expectation{Z_{2:n}} = \sum_{i=2}^n \frac{1}{i} = H_n - 1$, where $H_n$ is the $n$-th harmonic number. For large $n$, $H_n$ can be approximated by $\ln n$. Therefore, expression~\ref{eq: revenue of vcg} being at least expression~\ref{eq: revenue of auction} is equivalent to :
\begin{gather*}
(\days-1) \cdot \expectation{ X_{2:n+c} } + n + c > (\days-1) \cdot \expectation{ X_{1:n} } \\
(\days-1) \cdot \left( H_{n+c} - 1 \right) + n + c > (\days-1) \cdot H_n \\
H_{n+c} - H_n + \frac{n+c}{\days-1} > 1 \\
\ln \left( \frac{n+c}{n} \right) + \frac{n+c}{\days-1} > 1 \\
(n+c) e^{\frac{n+c}{\days-1}} > n e \\
\frac{n+c}{\days-1} e^{\frac{n+c}{\days-1}} > \frac{n e}{\days-1} \\
\frac{n+c}{\days-1} > W \left( \frac{n e}{\days-1} \right) \\
c > (\days-1) \cdot W \left( \frac{n e}{\days-1} \right) - n \\
\end{gather*}
where W is the Lambert function. $\lim_{k \rightarrow \infty} k W (\frac{en}{k}) = en$, therefore $c > (e-1)n$.
\end{proof}


\begin{lemma}
For $\days$ MHR stages, and ex-ante IR auctions, the Competition Complexity is at least $(e-1)n$, even for independent stages and for auctions that are incentive compatible in a perfect Bayesian equilibrium.
\end{lemma}

\begin{proof}
For the first stage, $X_1$ is uniform $[0,\epsilon]$ for some small $\epsilon > 0$. For $k=2,\dots,\days$, $X_k = Exp(1,V)$ is an exponential distribution with parameter $\lambda = 1$, truncated at some large $V$. We describe an ex-ante IR auction that is incentive compatible in a perfect Bayesian equilibrium, with revenue $\sum_{k=2}^\days \expectation{ (X_k)_{1:n} }$, i.e. essentially the social welfare: on stages $2$ through $\days$ the auctioneer will run a second price auction, extracting revenue $\sum_{k=2}^\days \expectation{ (X_k)_{2:n} }$. The extra $\sum_{k=2}^\days \left( \expectation{ (X_k)_{1:n} } - \expectation{ (X_k)_{2:n} } \right)$ will be charged upfront; in stage $1$ every buyer is offered the option to pay $\frac{1}{n} \sum_{k=2}^\days \left( \expectation{ (X_k)_{1:n} } - \expectation{ (X_k)_{2:n} } \right)$ in order to participate in stages $2$ through $k$. This is equal to the expected utility of each buyer, and therefore, since they are expectation maximizers, they will accept the offer.

Given this auction, the calculation for lower bounding the Competition Complexity is almost identical (in fact much simpler) to Lemma~\ref{lemma: lower bound on cc for independent }.  
\end{proof}
\section{Proofs missing from Section~\ref{sec:second price bounds}}\label{appendix:missing from mhr}

\begin{proof}[Proof of Claim~\ref{claim: basic MHR }]
$\frac{d}{dx} \log\left( 1- F(x) \right) = \frac{-f(x)}{1-F(x)} = -h(x)$. Therefore, $1-F(x) = e^{- \int_0^x h(z) dz}$. Re-arranging proves the first part of the claim. We get the second part from the definition of expectation for non-negative random variables.
\end{proof}


\begin{proof}[Proof of Claim~\ref{claim:expected second max formula}]
\begin{align*}
\expectation{X_{2:4}} &= \int_0^{\bar{V}} 1 - F_{2:4}(x) dx \\
&= \int_0^{\bar{V}} 1 - 4 F(x)^{3} + 3 F(x)^4 dx \\
&= \int_0^{\bar{V}} 1 - 4 \left( 1 - e^{-H(x)} \right)^3 + 3 \left( 1 - e^{-H(x)} \right)^4 \\
&= \int_0^{\bar{V}} 3 e^{-4H(x)} - 8 e^{-3H(x)} + 6 e^{-2H(x)} dx
\end{align*}	
\end{proof}

\begin{proof}[Proof of Lemma~\ref{lemma: piecewise linear H}]
For a general piecewise linear and convex function $\hat{H}(x)$ with $c$ linear pieces, we have
\[ 
\textstyle \hat{H}(x) = 
\begin{cases} 
a_0 x + b_0 &\mbox{if }  x_1 \geq x \geq x_0 \\ 
a_1 x + b_1 & \mbox{if } x_2 \geq x \geq x_1 \\
\dots \\
a_c x + b_c & \mbox{if } x_{c+1} \geq x \geq x_c \\
\end{cases} 
\]

where $b_0 = 0$ and $a_i x_i + b_i= a_{i-1} x_i + b_{i-1}$, $\forall i \geq 1$. $x_0 = 0$, $x_{c+1} = \bar{V}$ and $x_{i+1} > x_i$ for all $i \geq 0$. Since $\hat{H}(x)$ is convex, $a_{i+1} \geq a_i > 0$, for all $i$. Let $I = \int_0^{\bar{V}} 3 e^{-4\hat{H}(x)} - 8 e^{-3\hat{H}(x)} + 6 e^{-2\hat{H}(x)} - e^{-\hat{H}(x)} dx$.

\begin{align*}
I &= \sum_{i=0}^c \int_{x_i}^{x_{i+1}} 3 e^{-4\left( a_i x + b_i \right) } - 8 e^{-3\left( a_i x + b_i \right) } + 6 e^{-2\left( a_i x + b_i \right) } - e^{-\left( a_i x + b_i \right) } dx  \\
&= \sum_{i=0}^c - \left[\frac{3e^{-4\left( a_i x + b_i \right)}}{4a_i}\right]_{x_i}^{x_{i+1}} + \left[\frac{8e^{-3\left( a_i x + b_i \right)}}{3a_i}\right]_{x_i}^{x_{i+1}} - \left[ \frac{3e^{-2\left( a_i x + b_i \right)}}{a_i}\right]_{x_i}^{x_{i+1}} + \left[\frac{e^{-\left( a_i x + b_i \right)}}{a_i}\right]_{x_i}^{x_{i+1}} \\
&= \sum_{i=0}^c \frac{1}{a_{i}} \left( -\frac{3e^{-4\left( a_i x_{i+1} + b_i \right)}}{4} + \frac{8e^{-3\left( a_i x_{i+1} + b_i \right)}}{3} - 3e^{-2\left( a_i x_{i+1} + b_i \right)} + e^{-\left( a_i x_{i+1} + b_i \right)} \right) \\
&\qquad - \sum_{i=0}^c \frac{1}{a_i}\left( - \frac{3e^{-4\left( a_i x_i + b_i \right)}}{4} + \frac{8e^{-3\left( a_i x_i + b_i \right)}}{3} - 3e^{-2\left( a_i x_i + b_i \right)} + e^{-\left( a_i x_i + b_i \right)} \right) \\
&= \sum_{i=0}^c \frac{1}{a_{i}} \left( -\frac{3e^{-4\left( a_i x_{i+1} + b_i \right)}}{4} + \frac{8e^{-3\left( a_i x_{i+1} + b_i \right)}}{3} - 3e^{-2\left( a_i x_{i+1} + b_i \right)} + e^{-\left( a_i x_{i+1} + b_i \right)} \right) \\
&\qquad + \frac{1}{12 a_0} - \sum_{i=1}^c \frac{1}{a_i}\left( - \frac{3e^{-4\left( a_i x_i + b_i \right)}}{4} + \frac{8e^{-3\left( a_i x_i + b_i \right)}}{3} - 3e^{-2\left( a_i x_i + b_i \right)} + e^{-\left( a_i x_i + b_i \right)} \right). \\
\end{align*}

For all $i\geq 1$ we have that $a_i x_i + b_i = a_{i-1} x_i + b_{i-1}$.
\begin{talign*}
I &= \sum_{i=0}^c \frac{1}{a_{i}} \left( -\frac{3e^{-4\left( a_i x_{i+1} + b_i \right)}}{4} + \frac{8e^{-3\left( a_i x_{i+1} + b_i \right)}}{3} - 3e^{-2\left( a_i x_{i+1} + b_i \right)} + e^{-\left( a_i x_{i+1} + b_i \right)} \right) \\
&\quad + \frac{1}{12 a_0} - \sum_{i=1}^c \frac{1}{a_i}\left( - \frac{3e^{-4\left( a_{i-1} x_i + b_{i-1} \right)}}{4} + \frac{8e^{-3\left( a_{i-1} x_i + b_{i-1} \right)}}{3} - 3e^{-2\left( a_{i-1} x_i + b_{i-1} \right)} + e^{-\left( a_{i-1} x_i + b_{i-1} \right)} \right) \\
&= \sum_{i=0}^c \frac{1}{a_{i}} \left( -\frac{3e^{-4\left( a_i x_{i+1} + b_i \right)}}{4} + \frac{8e^{-3\left( a_i x_{i+1} + b_i \right)}}{3} - 3e^{-2\left( a_i x_{i+1} + b_i \right)} + e^{-\left( a_i x_{i+1} + b_i \right)} \right) \\
&\quad  \frac{1}{12 a_0} - \sum_{i=0}^{c-1} \frac{1}{a_{i+1}}\left( - \frac{3e^{-4\left( a_{i} x_{i+1} + b_{i} \right)}}{4} + \frac{8e^{-3\left( a_{i} x_{i+1} + b_{i} \right)}}{3} - 3e^{-2\left( a_{i} x_{i+1} + b_{i} \right)} + e^{-\left( a_{i} x_{i+1} + b_{i} \right)} \right) \\
&=  \frac{1}{12 a_0} + \frac{1}{a_{c}} \left( -\frac{3e^{-4\left( a_c x_{c+1} + b_c \right)}}{4} + \frac{8e^{-3\left( a_c x_{c+1} + b_c \right)}}{3} - 3e^{-2\left( a_c x_{c+1} + b_c \right)} + e^{-\left( a_c x_{c+1} + b_c \right)} \right) \\
&+ \sum_{i=0}^{c-1} \left( \frac{1}{a_{i}} - \frac{1}{a_{i+1}} \right)  \left( -\frac{3e^{-4\left( a_i x_{i+1} + b_i \right)}}{4} + \frac{8e^{-3\left( a_i x_{i+1} + b_i \right)}}{3} - 3e^{-2\left( a_i x_{i+1} + b_i \right)} + e^{-\left( a_i x_{i+1} + b_i \right)} \right). \\
\end{talign*}

We add and subtract $\sum_{i=1}^{c} \frac{1}{12 a_i}$:

\begin{talign*}
I &=  \frac{1}{a_{c}} \left( -\frac{3e^{-4\left( a_c x_{c+1} + b_c \right)}}{4} + \frac{8e^{-3\left( a_c x_{c+1} + b_c \right)}}{3} - 3e^{-2\left( a_c x_{c+1} + b_c \right)} + e^{-\left( a_c x_{c+1} + b_c \right)} \right) \\
&\quad + \frac{1}{12 a_0} + \frac{1}{12a_1} + \frac{1}{12a_2} + \dots + \frac{1}{12a_c} - \frac{1}{12a_1} - \frac{1}{12a_2} - \dots - \frac{1}{12a_c} \\
&\quad + \sum_{i=0}^{c-1} \left( \frac{1}{a_{i}} - \frac{1}{a_{i+1}} \right)  \left( -\frac{3e^{-4\left( a_i x_{i+1} + b_i \right)}}{4} + \frac{8e^{-3\left( a_i x_{i+1} + b_i \right)}}{3} - 3e^{-2\left( a_i x_{i+1} + b_i \right)} + e^{-\left( a_i x_{i+1} + b_i \right)} \right) \\
&= \frac{1}{a_{c}} \left( -\frac{3e^{-4\left( a_c x_{c+1} + b_c \right)}}{4} + \frac{8e^{-3\left( a_c x_{c+1} + b_c \right)}}{3} - 3e^{-2\left( a_c x_{c+1} + b_c \right)} + e^{-\left( a_c x_{c+1} + b_c \right)} \right) \\
&\quad + \left( \frac{1}{12 a_0} - \frac{1}{12a_1} \right) + \left( \frac{1}{12a_1} - \frac{1}{12a_2} \right) + \dots + \left( \frac{1}{12a_{c-1}} - \frac{1}{12a_c} \right) + \frac{1}{12a_c} \\
&\quad + \sum_{i=0}^{c-1} \left( \frac{1}{a_{i}} - \frac{1}{a_{i+1}} \right)  \left( -\frac{3e^{-4\left( a_i x_{i+1} + b_i \right)}}{4} + \frac{8e^{-3\left( a_i x_{i+1} + b_i \right)}}{3} - 3e^{-2\left( a_i x_{i+1} + b_i \right)} + e^{-\left( a_i x_{i+1} + b_i \right)} \right) \\
&= \frac{1}{a_{c}} \left( -\frac{3e^{-4\left( a_c x_{c+1} + b_c \right)}}{4} + \frac{8e^{-3\left( a_c x_{c+1} + b_c \right)}}{3} - 3e^{-2\left( a_c x_{c+1} + b_c \right)} + e^{-\left( a_c x_{c+1} + b_c \right)} + \frac{1}{12} \right) \\
& + \sum_{i=0}^{c-1} \left( \frac{1}{a_{i}} - \frac{1}{a_{i+1}} \right)  \left( -\frac{3e^{-4\left( a_i x_{i+1} + b_i \right)}}{4} + \frac{8e^{-3\left( a_i x_{i+1} + b_i \right)}}{3} - 3e^{-2\left( a_i x_{i+1} + b_i \right)} + e^{-\left( a_i x_{i+1} + b_i \right)} + \frac{1}{12} \right). \\
\end{talign*}

Let $$g(y) = -\frac{3e^{-4y}}{4} + \frac{8e^{-3y}}{3} - 3e^{-2y} + e^{-y} + \frac{1}{12}.$$ Taking the derivative gives $g'(y) = 3e^{-4y} -8e^{-3y} + 6e^{-2y} - e^{-y} $. Solving for $g'(y) = 0$ for $y \geq 0$ when $y=0, \ln \left( \frac{5+\sqrt{13}}{2} \right)$ and as $y$ goes to infinity.
For $y=0$, $g(0) = 0$.
$lim_{y \rightarrow \infty} g(y) = \frac{1}{12}$.
For $y=\ln \left( \frac{5+\sqrt{13}}{2} \right)$, 
$$ g \left( \ln \left( \frac{5+\sqrt{13}}{2} \right) \right) = \frac{1}{12} - \frac{12}{(5 + \sqrt{13})^4} + \frac{64}{3 (5 + \sqrt{13})^3} - \frac{12}{(5 + \sqrt{13})^2} + \frac{2}{5 + \sqrt{13}} > 0. $$

Since $\frac{1}{a_{i}} - \frac{1}{a_{i+1}}$, $\frac{1}{a_{c}}$, $a_c x_{c+1} + b_c$ and $a_i x_{i+1} + b_i$ are all strictly positive, $I > 0$.
\end{proof}

\begin{proof}[Proof of Lemma~\ref{lemma: piecewise linear approximation}]
Since $f(x)$ is continuous, we can pick a $\delta(\epsilon) = \delta > 0$, such that for any $x,x' \in [a,b]$ such that $\left| x - x' \right| < \delta$, we have $\left| f(x) - f(x') \right| < \epsilon$. Let $n$ be an integer such that $\frac{b-a}{n} < \delta$, and let $x_j = a + j \cdot \frac{b-a}{n}$, for $j = 0,1,\dots,n$. On each $[x_j,x_{j+1}]$ define $g_{\epsilon}(x)$ to be the function whose graph is the line segment connecting $\left( x_j,f(x_j) \right)$ with $\left( x_{j+1},f(x_{j+1}) \right)$, i.e. $g_{\epsilon} (x) = f(x_{j}) + t \cdot \left( f(x_{j+1}) - f(x_j)  \right)$ for some $t \in [0,1]$. Since $g_{\epsilon} (x) \in [f(x_j), f(x_{j+1})]$, the Intermediate Value Theorem implies that for all $x \in [x_j, x_{j+1}]$ there exists a $y_j \in [x_j,x_{j+1}]$, such that $g_{\epsilon}(x) = f(y_j)$. Since each $x \in [a,b]$ belongs in some $[x_j,x_{j+1}]$, the choice of $\delta$ implies that $ \left| f(x) - g_{\epsilon}(x) \right| = \left| f(x) - f(y_j) \right| < \epsilon $. It remains to show that $g_{\epsilon}(x)$ is convex: by construction, the slope of $g_{\epsilon}(x)$ for $x \in [x_j,x_{j+1}]$ is $\left( f(x_{j+1}) - f(x_j) \right) / \left( x_{j+1} - x_j \right)$, which is strictly increasing with $j$ since $f(x)$ is convex.
\end{proof}

\begin{proof}[Proof of Lemma~\ref{lem:expected min of two samples}]
\begin{align*}
\expectation{X_{2:2}} &= \int_0^{\infty} (1-F(x))^2 dx = \int_0^{\infty} e^{-2H(x) } dx \\
&\geq \int_0^{\infty} e^{-H(2x)} dx \\
&= \frac{1}{2} \int_0^{\infty} 1-F(x) dx = \frac{1}{2} \expectation{X},
\end{align*}
where the second line holds because $H(x)$ is a convex function. Since $\expectation{X_{1:2}} + \expectation{X_{2:2}} = 2 \expectation{X}$, we have that $\expectation{X_{1:2}} = 2 \expectation{X} - \expectation{X_{2:2}} \leq \frac{3}{2} \expectation{X}$. Therefore, $\expectation{X_{2:2}} \geq \frac{1}{2} \expectation{X} \geq \frac{1}{2} \frac{2}{3}\expectation{X_{1:2}} = \frac{1}{3} \expectation{X_{1:2}}$. The Lemma follows.
\end{proof}

\begin{proof}[Proof of Lemma~\ref{lem: expected spacing is decreasing}]
\begin{align*}
\expectation{X_{1:n}} - \expectation{X_{2:n}} &= \int_0^{\infty} F_{2:n}(x) - F_{1:n}(x) dx\\
&=\int_0^{\infty} nF^{n-1}(x) - (n-1) F^n(x) - F^n(x) dx \\
&=\int_0^{\infty} n F^{n-1}(x) \left( 1 - F(x) \right) dx \\
&=\int_0^{\infty} n F^{n-1}(x) f(x) \frac{1}{h(x)} dx \\
&=\int_0^{\infty} f_{1:n}(x) \frac{1}{h(x)} dx \\
&= \expectation{ \frac{1}{h(X_{1:n})} } \qedhere\\
\end{align*}
\end{proof}




\end{appendix}

\end{document}